\documentclass[]{article} %

\usepackage{amsmath}
\usepackage{amssymb}
\usepackage{amsfonts}
\usepackage{amsthm}
\usepackage{graphicx}
\usepackage{hyperref}
\usepackage[linesnumbered, nofillcomment, noend]{algorithm2e} 
\usepackage{graphicx} 
\usepackage{microtype}   %
\usepackage[]{geometry}  %
\usepackage[english]{babel}
\usepackage{hyperref}  %

\newtheorem{thm}{Theorem}
\newtheorem{lem}[thm]{Lemma}

\newtheorem{definition}[thm]{Definition}
\newtheorem{problem}[thm]{Problem}

\DeclareRobustCommand{\NC}{\ensuremath{\mathrm{NC}}}
\DeclareRobustCommand{\P}{\ensuremath{\mathrm{P}}}

\DeclareRobustCommand{\NL}{\ensuremath{\mathrm{NL}}}

\DeclareRobustCommand{\lineseg}[3]{\ifthenelse{\equal{#2}{}}{\ensuremath{[#1]}}{\ensuremath{[#1, #2]_{#3}}}}
\DeclareRobustCommand{\nb}[1]{\ensuremath{\widetilde{#1}}}

\begin{document}
\title{\Large Parallel computation using active self-assembly\thanks{
    \url{mpchen@caltech.edu}, \url{dorx@alumni.caltech.edu}, \url{woods@caltech.edu}.     Supported by National Science Foundation grants CCF-1219274, 0832824 (The Molecular Programming Project), and CCF-1162589.
Preliminary version appeared at The 19th International Conference on DNA Computing and Molecular Programming (DNA 19).} } 
\author{Moya Chen
   \hspace{0.1\textwidth}
   Doris Xin
    \hspace{0.1\textwidth} 
 Damien Woods  
 }

\date{California Institute of Technology %
      }
\maketitle

\begin{abstract} 
We study the computational complexity of the recently proposed nubot model of molecular-scale self-assembly. The model generalises asynchronous cellular automata to have non-local movement where large assemblies of molecules can be pushed and pulled around, analogous to millions of molecular motors in animal muscle effecting the rapid movement of macroscale arms and legs. We show that the nubot model is capable of simulating Boolean circuits of polylogarithmic depth and polynomial size, in only polylogarithmic expected time. In computational complexity terms, we show that any problem from the complexity class NC is solvable in polylogarithmic expected time and polynomial workspace using nubots. 

Along the way, we give  fast parallel nubot algorithms for a number of problems including line growth, sorting, Boolean matrix multiplication and space-bounded Turing machine simulation, all using a constant number of nubot states (monomer types). Circuit depth is a well-studied notion of parallel time, and our result implies that the nubot model is a highly parallel model of computation in a formal sense. Asynchronous cellular automata are not capable of this parallelism, and our result shows that adding a rigid-body movement primitive to such a model, to get the nubot model, drastically increases parallel processing abilities.
\end{abstract}

\section {Introduction}\label{sec:intro}
We study the theory of molecular self-assembly, working within the recently-introduced {\em nubot} model   by Woods, Chen, Goodfriend, Dabby, Winfree and Yin~\cite{nubots}. Do we really need another new model of self-assembly? 
Consider the biological process of embryonic  development: a single cell growing into an organism of astounding  complexity. 
Throughout this active, fast and robust process there is growth and movement. For example, at an early stage in the development of the fruit fly Drosophila, the embryo contains approximately 6,000 large cells arranged on its ellipsoid-shaped surface. 
Then, in just four minutes, the embryo rapidly changes  shape to become invaginated, creating a large structure  that becomes the mesoderm, and ultimately  muscle. How does this fast rearrangement occur? 
A large fraction of these cells undergo a rapid, synchronised and highly parallel rearrangement of their internal structure where, in each cell, one end of the cell contracts and  the other end expands. This is achieved by a mechanism that seems to crucially involve thousands of molecular-scale myosin motors  pulling and pushing the cellular cytoskeleton to quickly effect this rearrangement~\cite{Wieschaus08pulsed}. At an abstract level one can imagine this as being analogous to how millions of molecular motors in a muscle, each taking a tiny step but acting in a highly parallel fashion, effect rapid long-distance muscle contraction. This rapid parallel movement, combined with the constraint of a fixed cellular volume, as well as variations in the elasticity properties of the cell membrane, can explain this key step in embryonic morphogenesis. Indeed,  molecular motors that together, in parallel, produce macro-scale movement are a ubiquitous phenomenon in biology.  

We wish to understand, at a high level of abstraction, the ultimate limitations and capabilities of such molecular scale rearrangement and growth.  We do this by studying the computational power of a theoretical  model that includes these capabilities.  As a first step towards such understanding, we show in this paper that large numbers of tiny motors (that can each  pull or push a tiny amount) coupled with local state changes on a grid, are sufficient to quickly solve  inherently parallelisable problems. This result, described formally below in Section~\ref{sec:results},  demonstrates that the nubot model is a highly parallel computer in a computational complexity-theoretic sense.

Another motivation, and potential test-bed for our theoretical model and results, is the fabrication of active molecular-scale structures.  Examples  include DNA-based walkers, DNA origami that  reconfigure,  and simple structures called molecular motors~\cite{yurke00} that  transition between a small number of  discrete states (see~\cite{nubots} for references). In these systems the interplay between structure and dynamics  leads to behaviours and capabilities that are not seen in static structures, nor in other unstructured but active, well-mixed chemical reaction network type  systems.  Our theoretical results here, and those in~\cite{nubots}, provide a sound basis to motivate the experimental investigation of large-scale active DNA nanostructures. 

There are  a number of theoretical models of  molecular-scale algorithmic self-assembly processes~\cite{patitz2012introduction}. For example, the abstract Tile Assembly Model, where individual square DNA tiles attach to a growing assembly lattice one at a time~\cite{Winf98thesis,winfree00rot, IUSA}, the two-handed (hierarchical) model where large multi-tile assemblies come together~\cite{Aggarwal_etal2005generalized,cannon2012two,demaine2008staged,2HAMIU}, and the signal tile model where DNA origami tiles that form an ``active'' lattice with DNA strand displacement signals running along them~\cite{jonoska2012active,padilla2011hierarchical, padilla2012asynchronous}. Other models enable one to  program tile geometry~\cite{one,fu2011self},  temperature~\cite{Aggarwal_etal2005generalized,   kao2006reducing,summers2012temperature}, concentration~\cite{becker2010,chandran2012tile,doty2010concShort,kao2008conc}, mixing stages~\cite{demaine2008staged,Staged1D_DNA2011}  and connectivity/flexibility~\cite{jonoska2009complexity}. 

The well-studied abstract Tile Assembly Model~\cite{Winf98thesis} is an asynchronous, and nondeterministic, cellular automaton with the restriction that state changes are irreversible and happen only along a crystal-like growth frontier. The nubot model is a generalisation of an asynchronous and nondeterministic cellular automaton, where the generalisation is that we have a non-local movement primitive.   Since the nubot model is intended to be a model of molecular-scale phenomena it ignores friction and gravity, allows for the creation/destruction of monomers (we assume an invisible ``fuel'' source) and has a notion of random uncontrolled motion (called agitation, but not used in this paper). Instances of the model evolve as continuous time Markov processes, and time is modelled as in stochastic chemical kinetics~\cite{gillespie1992rigorous, soloveichik2008computation}. The nubot style of rigid-body movement is analogous to that seen in reconfigurable robotics~\cite{rus02but, rus2001crystalline,  murata2007self}, and indeed results in these robotics models show that non-local movement can be used to effect fast global reconfiguration~\cite{Crystalline_ISAAC2008,aloupis2011efficient,reif2007optimal}.  The nubot model includes features seen in cellular automata, Lindenmayer systems~\cite{prusinkiewicz1991algorithmic}  and graph grammars~\cite{klavins04kla}. See~\cite{nubots} for a more detailed comparison of the similarities and differences with these models.  

\subsection{Previous work on active self-assembly with movement}
Previous work on the nubot model~\cite{nubots} showed that it  is capable of building large  shapes and patterns exponentially quickly:  e.g.\ lines and squares in time logarithmic in their size. The same paper goes on to describe a general scheme to build arbitrary computable (connected, 2D) size-$n$ shapes in time and number of monomer  states (types) that are  polylogarithmic in $n$, plus the time and states required for Turing machine simulation due to  the inherent algorithmic complexity of the shape.   Furthermore, 2D patterns with at most~$n$ coloured pixels, where the colour choice for each pixel is  computable in time $\log^{O(1)} n$ (i.e. polynomial in the length of the binary description of pixel indices),  are nubot-computable in time and number  of monomer types polylogarithmic in~$n$~\cite{nubots}. The latter result is achieved without going outside the pattern boundary and in a completely asynchronous fashion. These results show that the nubot model is capable of parallelism not seen in many other models of self-assembly. The goal of the present paper is to characterise the kind of parallelism seen in the nubot model by formally relating it to the computational complexity of classical decision~problems. 

Dabby and Chen~\cite{dabbyChenSODA2012} study a 1D model, where  monomers insert between, and push apart, other monomers. Their model is closely related to a 1D restriction of the nubot model without state changes, and they build length $n$ lines in $O(\log^3 n)$ expected time and $O(\log^2 n)$ monomer types.  They also show that the set of 1D polymers produced by any instance of their model is a context-free language, and give a  design for implementation with DNA molecules.  
Malchik and Winslow~\cite{MalchikWinsow} first show that any context-free language can be expressed as an instance of this model, and then give an asymptotically tight bound of $2^{\Theta{(k^{3/2})}}$ on the length of polymers produced using $k$ monomer types (in merely  $O(\log^{5/3} n)$ expected time), thus  characterising two aspects of the model.

\subsection{Main result}\label{sec:results}
In the nubot model a program is specified as a finite set of nubot rules $\mathcal{N}$ and is said to decide a language $L \subseteq \{0,1 \}^*$ if, beginning with a word $x \in \{0,1 \}^*$ encoded as a sequence of $|x|$ ``binary monomers'', the system eventually reaches a configuration containing exactly the 1 monomer if $x\in L$, and 0 otherwise.    
Let $\NC$ denote the (well-known) class of problems solved by uniform polylogarithmic depth and polynomial size Boolean circuits.\footnote{$\NC$, or Nick's class, is named after Nicholas Pippenger.} Our main result is stated as follows. 

\begin{thm}\label{thm:nubotsSolveNC} For each language $L \in \NC$, there is a set of nubot rules $\mathcal{N}_L$ that decides~$L$ in polylogarithmic expected time, constant number of monomer states, and polynomial space in the input string  length. Moreover, for $i \geq 1$, $\NC^i$ is contained in the class of languages decided by nubots running in $O(\log^{i+3} n)$ expected time, $O(1)$ monomer states, and polynomial space in  input  length $n$.
\end{thm}

$\NC$ problems are solved by circuits of shallow depth, hence they can be thought of as those problems that can be solved on a highly parallel architecture (simply run each layer of the circuit on a bunch of parallel processors, after polylogarithmic parallel steps we are done).   
 $\NC$ is contained in~$\P$---problems solved by polynomial time Turing machines---and this follows from the fact that $\NC$ circuits are of polynomial size.  Problems in $\NC$, and the analogous function class, include sorting, Boolean matrix multiplication, various kinds of maze solving and graph reachability, and integer addition, multiplication and division. Besides its circuit depth definition, $\NC$ has been characterised by a large number of other parallel  models of computation including parallel random access machines,  vector machines, and optical computers~\cite{greenlaw1995limits,woods2008parallel,woods2005upper}. 
It is widely conjectured, but unproven, that  $\NC$  is strictly contained in  $\P$. In particular,  problems complete for $\P$ (such as Turing machine and cellular automata~\cite{nearyWoods2006rule110} prediction, context-free grammar membership and many others~\cite{greenlaw1995limits}) are believed to be ``inherently sequential''---it is conjectured that these problems are not solvable by parallel computers that run for polylogarithmic time on  a polynomial number of processors~\cite{greenlaw1995limits,condon1994theory}. 

Thus our main result gives a formal sense in which the nubot model is highly parallel: for any highly parallelisable ($\NC$) problem  our proof gives a nubot algorithm  to efficiently solve in it in only polylogarithmic expected time and constant states.  This stands in contrast to sequential machines like Turing machines, that cannot read all of an $n$-bit input string in polylogarithmic time, and  ``somewhat parallel'' models like cellular automata and the abstract Tile Assembly Model, which can not have all of $n$ bits influence a single bit decision in polylogarithmic time~\cite{keenan2014fast}.  Thus, adding a movement primitive to an asynchronous non-deterministic cellular automation, as in the nubot model, drastically increases its parallel processing abilities. 

We finish this discussion on a technical remark. Previous results~\cite{nubots} on the nubot model were of the form: for each $n \in \mathbb{N}$ there is a set of nubot rules $\mathcal{N}_n$ (i.e. the number of rules is a  function of~$n$)  to carry out some task parameterised by $n$ (examples: quickly grow a line of length~$n$ or an $n \times n$ square, or grow some complicated computable pattern or shape whose size is parameterised by~$n$, etc.).  For each problem in $\NC$ our main result here gives a {\em single} set of rules (i.e.\ of constant size), that works for all problem instances.

\subsection{Overview of results and paper structure}
Section~\ref{sec:intro} contains the statement of our main result,  the overall proof structure and some future work directions. Section~\ref{sec:defs} gives the full definition of the nubots model and relevant complexity classes. Section~\ref{sec:linedouble} serves as an introduction to the nubots model by giving a simple nubots algorithm to double the length of a length-$n$ line in $O(\log n)$ expected time, we suggest the reader begins there. 

\subsubsection{New synchronization and line growth algorithms}
  In Section~\ref{sec:sync} we describe a fast signalling method for nubots from~\cite{nubots}, here called {\em shift synchronization}, and give a new variant on this called {\em lift synchronization}. These signalling mechanisms are used throughout our constructions as a method to quickly send a bit, 0 or 1, distance $n$ in $O(\log n)$ expected time, with the choice of 0 or 1 being encoded by the use of shift or lift synchronization respectively. 

The line growth algorithm given in \cite{nubots} grows a line of length $n\in\mathbb{N}$ in $O(\log n)$ time, using $O(\log n)$ monomer states and starting from a single monomer on the grid.  Section~\ref{sec:line} gives a new line-growth  algorithm that completes in $O(\log^2 n)$ time, using $O(1)$ monomer states and starting from $O(\log n)$  monomers on the grid.  A key feature of our algorithm is that it uses only a constant number of states. This helps us achieve our main result, which requires a {\em single} set of nubots rules that accept any word from some, possibly infinite, $\NC$ language: as part of our circuit simulation we need to build longer and longer lines to simulate larger and larger circuits, all with a single set of nubots rules.

\subsubsection{Parallel sorting, Boolean matrix multiplication \& space bounded Turing machine simulation} 
Section~\ref{sec:sorting} shows that the nubots model is capable of fast parallel sorting: $n$ numbers can be sorted in expected  time polylogarthmic in $n$.  More precisely, $n$ distinct natural numbers, taken from the set $\{0,1,\ldots , n-1 \}$ %
when presented as $n$ unordered ``strings'' of binary (0 or 1) monomers on the grid, can be sorted in increasing numerical order in expected time $O(\log^3 n )$, space $O(n \log n) \times O(n)$, and  $O(1)$ monomer states.
Our sorting routine is used throughout our main construction and is inspired by mechanisms, such as gel electrophoresis, that sort via spatial organization based on physical quantities, such as mass and charge~\cite{MurphySort06}.  

Section~\ref{sec:matrixmult} shows that two $n \times n$  Boolean matrices can be  multiplied  in $O(\log^3 n)$ expected time, $O(n^4 \log n) \times O(n^2 \log n)$ space and $O(1)$ monomer states. This immediately implies that problems reducible to Boolean matrix multiplication, such as directed graph reachability and indeed any problem in the complexity class $\NL$, of  languages accepted by nondeterministic logarithmic space bounded Turing machines, can be solved in polylogarithmic expected time on nubots.

Indeed in Section~\ref{sec:TMsim} we go on to generalise this result by showing that any  nondeterministic logarithmic space bounded Turing machine that computes a function (as opposed to merely deciding a language) can also be simulated in polylogarithmic space. This involves modifying the usual matrix multiplication method to keep track of the contents of the output tape of the Turing machine, and correctly reassembling the encoded tape contents on the 2D grid. 

These results show that the model is capable of fast parallel solution of many problems, in particular all of those in $\NL$. Recall that $\NL \subseteq \NC$, so we are not done yet. Indeed these techniques form part of our more general result: polylogarithmic expected time solution of problems in $\NC$ via efficient simulation of uniform Boolean circuits, as described next.

\subsubsection{Proof overview of main result: Theorem~\ref{thm:nubotsSolveNC}}
Let  $L \in \NC$, in other words, $L$ is  decidable by a logspace-uniform family $\mathcal{C}_L$ of Boolean circuits of polylogarthmic depth and polynomial size. To prove Theorem~\ref{thm:nubotsSolveNC},  we show that for each such~$L$ there exists a finite set of nubots rules~$\mathcal{N}_L$  that  decides~$L$. 
$L $ being in logspace-uniform $\NC$ implies that there is a deterministic logarithmic space (in input size) Turing machine $\mathcal{M}_L $  such  that $\mathcal{M}_L(1^n) = c_n $, where~$c_n$ is a description of the unique Boolean circuit in $\mathcal{C}_L$ that has $n$ input gates.  Our initial nubots configuration consists of a length-$n$ line of binary nubots monomers denoted~$\lineseg{\nb{x}}{}{}$, that represents  some input word $x \in \{0,1\}^*$  (as described in Definition~\ref{def:nubotdecide}). From this we  create (copy) another length-$n$ line of monomers
that encode the unary string $1^n$ to be given as input to a nubots simulator of $\mathcal{M}_L$. 
The rule set $\mathcal{N}_L$ includes a description of $\mathcal{M}_L$, and the system first generates a circuit by simulating the computation of $\mathcal{M}_L$ on input $1^n$,  which produces a nubots configuration (collection of monomers in a connected component) that represents the circuit $c_n$. The circuit is then simulated on input~$x$.  Both of these tasks present a number of challenges. 

\paragraph{Circuit Generation.}
Logspace Turing machines run in at most polynomial time in their input length (otherwise they repeat a configuration), but here we wish to generate the circuit in merely polylogarithmic time. To achieve this, our simulation of~$\mathcal{M}_L$ works in a highly parallel fashion. This uses a number of techniques. First,  in nubots, we implement the (known) trick of space-bounded Turing machine simulation by fast iterated matrix multiplication, which in turn is used to solve reachability on the directed graph of all possible configurations of the Turing machine. One of  the main challenges here is to carry out matrix multiplication on the 2D grid sufficiently fast but without monomers unintentionally colliding with each other. Second, although iterated matrix multiplication is sufficient to simulate a Turing machine that decides a language,  here we wish to simulate a Turing machine that computes a {\em function}. To do this, our parallel matrix multiplication algorithm  keeps track of any symbols written to the output tape by both valid (reachable) and invalid (unreachable)  configurations, and at the end deletes those symbols written by invalid configurations leaving the valid output symbols only. These valid output symbols are then arranged into the correct order by our fast parallel sorting routine.  This results in a string of monomers that encodes the circuit~$c_n$. These monomers then rearrange themselves in the plane, to lay out the circuit with each row of gates layered one on top of the next as shown in Figure~\ref{fig:boolAndNubot}. (Note that for convenience and to save space we sometimes draw figures on a square grid, although the nubots model is formally defined on the hexagonal grid.) 

\paragraph{Circuit Simulation.}  As already described, the input $x$ is encoded as the binary monomers $\lineseg{\nb{x}}{}{}$, and the entire circuit $c_n$ is ``grown'' from $\lineseg{\nb{x}}{}{}$. The monomers $\lineseg{\nb{x}}{}{}$ now move to the first (bottom) row of the encoded circuit  (Figure~\ref{fig:boolAndNubot}(c)) and position themselves so that each gate can ``read'' its 1 or 2 input bit monomers from $\lineseg{\nb{x}}{}{}$.  After each gate computes a ``result'' bit,  layer~$1$ ``synchronizes'' via a $O(\log n)$ expected time synchronization routine. 

Next, we wish to send the ``result'' bits from layer~$1$ to layer~$2$.  Circuits are not necessarily planar, so we need to handle wire crossings.  We use our fast parallel sorting routine: the outputs from the first circuit layer are sorted, from left to right in increasing order, using their ``to'' address as a key. For example, a layer 1 result bit that is destined for gate 5 in layer 2 will be placed to the left of a layer 1 result bit that is destined for gate 6 in layer 2.   Using this sorting routine, the blue ``wire address'' regions in the circuit (Figure~\ref{fig:boolAndNubot}(d))  are sorted in increasing order from left to right, then appropriately padded with empty space in between (using counters), and are passed up to the next level. Layer 1 then destroys itself. The entire circuit is simulated, level by level, from bottom to top, in this manner. After the ``output gate'' monomer computes its output bit it  destroys itself, leaving a single monomer in state $\mathrm{output}_0$ or $\mathrm{output}_1$. 
No more rules are applicable and so the system has  halted with its answer. This completes the overview of the simulation. 

This overview  ignores many details. In particular  the nubots model is asynchronous, that is, rule updates happen independently as discrete events in continuous time with no two events happening at the same time (as in stochastic chemical kinetics).  The construction includes a large number of synchronization steps and signal passing to  ensure that all parts of the construction are  appropriately staged, but yet the construction is free to carry out many fast, asynchronous, parallel steps between these ``sequential'' synchronization steps.  

\begin{figure}[t]
  \centering
    \includegraphics[width=0.8\textwidth]{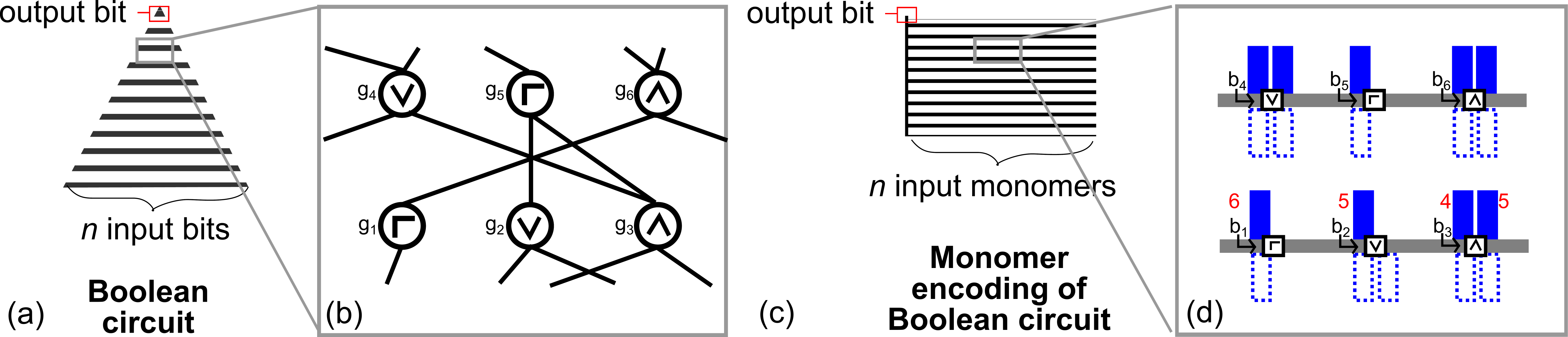}
  \caption{High-level overview of the encoding of a Boolean circuit as a nubots configuration (drawn on the square grid to save space).  (a) Boolean circuit with (b) detailed zoom-in. (c)  Nubots configuration encoding the circuit, with zoom-in shown in (d). A wire leading out of a gate in (b) has a destination gate number encoded in (d) as strips of $O(\log n)$ blue binary monomers (indices in red). After a gate computes some Boolean function (one of $\vee$, $\wedge$, $\neg$) the resulting bit is  tagged onto the relevant blue strip of monomers that encode the  destination addresses (red numbers). Circuits are not necessarily planar, so to handle wire crossovers these result bits are  first sorted in parallel based on their wire address, and then pushed up to the next layer of gates. }
  \label{fig:boolAndNubot}
\end{figure}

\subsection{Future work and open questions}
The line growth algorithm in~\cite{nubots} runs in expected time $O(\log n)$, uses $O(\log n)$ states  and space $O(n) \times O(1)$. In Section~\ref{sec:line} we give another  line growth algorithm that runs in expected time $O(\log^2 n)$, uses $O(1)$ states  and space $O(n) \times O(1)$.  Is there a line-growth algorithm that does better than time $\times$ space $\times$ states $= \Omega (n \log^2 n)$? To keep the game fair, the input should be a collection of monomers with space $\times$ states $= O (\log n)$.

Theorem~\ref{thm:nubotsSolveNC}  gives a lower bound on the power of the nubot model. What are the best lower and the upper bounds on the power of confluent\footnote{By confluent we mean a kind of determinism where the system (rules with the input) is assumed to always make a unique single terminal assembly.} polylogarithmic expected time nubots?  One challenge involves finding better Turing machine space, or circuit depth, bounds on computing multiple applications of the movable set (see Section~\ref{sec:defs}) on a polynomial size (or larger)  nubot grid.

Synchronization is a signalling method we use to quickly send signals in a non-local fashion. In this paper it is used extensively to compose nubot algorithms. What conditions are necessary and sufficient for composition of arbitrary nubot algorithms that do not use synchronization?  Theorem~7.1 in~\cite{nubots} shows that a wide class of patterns can be grown without synchronization, and its proof of this gives examples of composition without synchronization. 
It would be interesting to formalise this notion of composition in our distributed systems without the long-range fast signalling that synchronization gives.

Agitation is a kind of undirected, or random, movement that was defined for the nubot model in~\cite{nubots} and is intended to model a nanoscale environment where there are uncontrolled movements and turbulent fluid flows in all directions interacting with each monomer. Is it possible to simulate nubot-style movement using agitation? As motivation, note that every self-assembled molecular-scale structure was made under  conditions where agitation is a dominant source of movement! Our question asks if we can {\em programmably exploit} this random molecular motion to build structures quicker than without it.

Is the nubot model intrinsically universal? More precisely, does there exist a set of monomer rules $U$, such that any nubot system $\mathcal{N}$ can be simulated by ``seeding''   $U$ with a suitable initial configuration? The notion of intrinsic universality is giving rise to interesting characterisations, and separations, in a variety of tile assembly models~\cite{IUSA,USA, one, 2HAMIU,  temp1notIU,hendricks2013signal,HendricksPatitzTAMCA}, for an overview see the survey~\cite{WoodsIU}. Our hope would be that intrinsic universality, with its tight notion of simulation, could be used to tease apart the power of different notions of movement (for example to understand if  nubot-style movement  is weaker or stronger than other  notions of movement). 

 Other open problems and further directions can be found in~\cite{nubots}.

\newcommand{\config}{C}
\newcommand{\ra}{\rightarrow}
\newcommand{\la}{\leftarrow}
\newcommand{\La}{\Leftarrow}

\section{The nubot model and other definitions}\label{sec:defs}
In this section we formally define the nubot model. Figure~\ref{fig:model} gives an overview of the model and rules, and  Figure~\ref{fig:nondeterministicmovement} gives an example of the movement rule. An example nubot construction for ``line-doubling''  is given in Section~\ref{sec:linedouble} which may aid the reader at this point.  Let $\mathbb{N} = \{ 0,1,2,\ldots \}$.

\begin{figure}[t]
  \centering
    \includegraphics[width=\textwidth]{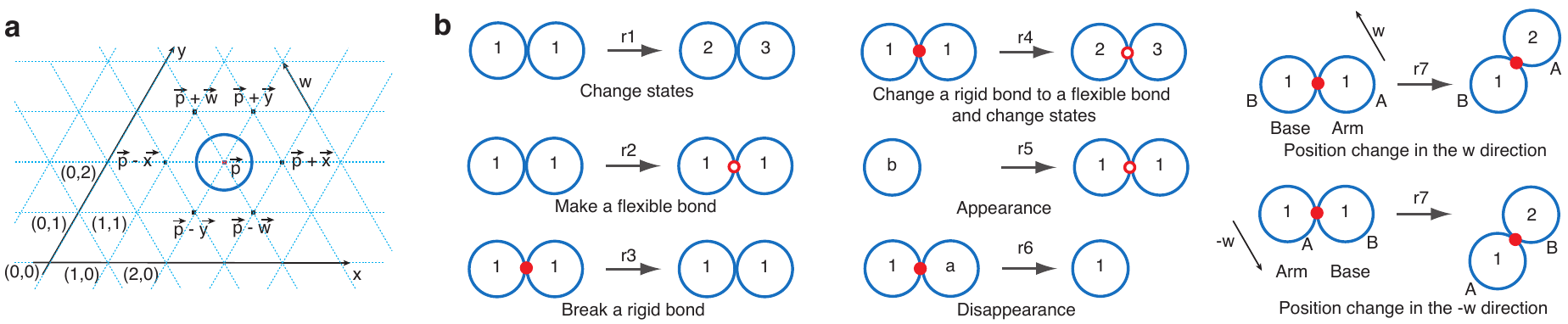}
  \caption{Overview of the nubot model. (a) A nubot configuration showing a single nubot monomer on the triangular grid. (b) Examples of nubot monomer rules. Rules r1-r6 are local cellular automaton-like rules, whereas r7 effects a non-local movement that may translate other monomers as shown in Figure~\ref{fig:nondeterministicmovement}. 
  A flexible bond is depicted as an empty red circle and a rigid bond is depicted as a solid red disk.}
    \label{fig:model}
\end{figure}
\begin{figure}[h]
  \centering
    \includegraphics[width=\linewidth]{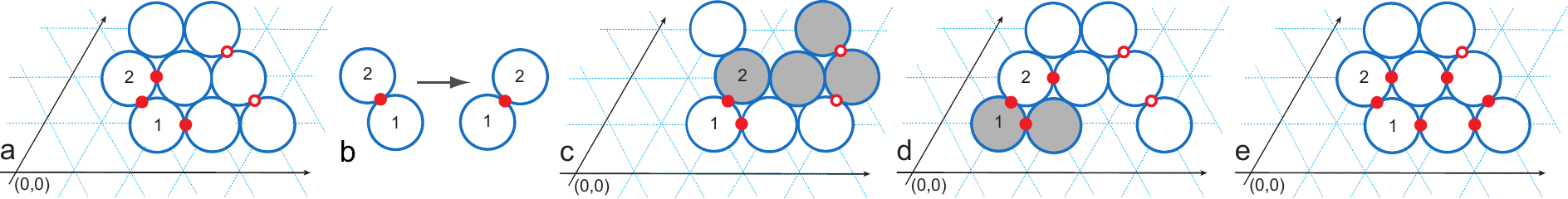}
  \caption{Movement rule. (a)~Initial configuration. (b)~Movement rule with one of two results depending on the (nondeterministic) choice of arm or base. (c) Result if the monomer with state 1 is the base, or (d) if the monomer with state 2 is the base. The shaded monomers are the moveable set.    We can think of (c) as pushing and (d) as pulling. The affect on rigid (filled red disks), flexible (hollow red circles) and null bonds is shown.  (e)~A configuration for which the movement rule is blocked: movement of 1 or 2 would force the other to move hence the rule is not applicable.}
  \label{fig:nondeterministicmovement}
\end{figure}

The model uses a two-dimensional triangular grid with a coordinate system using axes $x$ and~$y$ as shown in Figure~\ref{fig:model}(a). 
A third axis, $w$, is defined as running through the origin and $\overrightarrow{w} = -\overrightarrow{x} + \overrightarrow{y} = (-1,1)$, but we use only the~$x$ and~$y$ coordinates  to define position.
The \emph{axial directions} $\mathcal{D} =  \{ \pm\overrightarrow{x}, \pm\overrightarrow{y},  \pm\overrightarrow{w} \}$ are the unit vectors along  axes $x,y,w$.  A grid point $\overrightarrow{p} \in \mathbb{Z}^2$  has the set of six \emph{neighbours}  $\{ \overrightarrow{p} + \overrightarrow{u} \mid \overrightarrow{u} \in \mathcal{D} \}$.  Let $S$ be a finite set of monomer states. 
A nubot {\em  monomer} is a pair $X = (s_i, p(X)$) where $s_i \in S$ is a state and $p(X) \in \mathbb{Z}^2 $ is a grid point.
Two monomers on neighbouring grid points are either connected by a \emph{flexible}  or  \emph{rigid} bond, or else have no bond (called a \emph{null} bond).  Bonds are described in more detail below. 
A {\em configuration}~$C$ is a finite set of monomers along with the bonds between them. 

One configuration {\em transitions} to another via the application of a single \emph{rule}, $r = (s1, s2, b, \overrightarrow{u}) \rightarrow (s1', s2', b', \overrightarrow{u}')$ that acts on one or two monomers.\footnote{In reference~\cite{nubots} the nubot model includes ``agitation'': each monomer is repeatedly subjected to random movements intended to model a nano-scale environment where there is Brownian motion, uncontrolled movements and turbulent fluid flows in all directions. Our constructions in this paper work with or without agitation, hence they are robust to random uncontrolled movements, but we choose to ignore this issue and not formally define agitation for ease of presentation.}   The left and right  sides of the arrow respectively represent the contents of  two monomer positions before and after the application of rule~$r$.  
Here $s1, s2 \in S \cup \{ \mathsf{empty} \}$ are monomer states where at most one of $s1, s2$ is $ \mathsf{empty}$ (denotes lack of a monomer), 
$b \in \{\mathsf{flexible}, \mathsf{rigid}, \mathsf{null} \}$ is the bond type between them, and $\overrightarrow{u} \in \mathcal{D}$  is the relative position of the  $s2$ monomer  to the $s1$ monomer. 
If either of $s1$ or $s2$ (respectively $s1'$ or $s2'$) is $\mathsf{empty}$ then $b$ (respectively $b'$) is $ \mathsf{null}$  (monomers can not be bonded to empty space). The right is defined similarly, although there are some further restrictions on valid rules (involving $\overrightarrow{u}'$) described below. 
 A rule is only applicable in the orientation specified by~$\overrightarrow{u}$, and so rules are not rotationally invariant. 
 
 A rule may involve a movement (translation), or not.  First, in the case of no movement: $\overrightarrow{u} = \overrightarrow{u}'$. Thus we have a rule of the form $r = (s1, s2, b,\overrightarrow{u}) \rightarrow (s1', s2', b', \overrightarrow{u})$. From above,  at most one of $s1,s2$ is  
 $\mathsf{empty}$,  hence we disallow spontaneous generation of monomers from empty space. 
{\em State change} ($s1 \neq s1'$ and/or $s2 \neq s2'$ )  and {\em bond change} ($b \neq b'$) occur in a straightforward way,  examples are shown in Figure~\ref{fig:model}(b).
If $s_i \in \{s1,s2\}$ is $\mathsf{empty}$ and $s_i'$ is not, then the rule induces the \emph{appearance} of a new monomer  at the  empty location specified by $\overrightarrow{u}$ if $s2 = \mathsf{empty} $, or $-\overrightarrow{u}$ if $s1 = \mathsf{empty} $. If one or both monomer states go from non-empty to $\mathsf{empty}$, the rule induces the \emph{disappearance} of monomer(s) at the  orientation(s) given by $\overrightarrow{u}$. 

For a \emph{movement} rule, $\overrightarrow{u} \neq \overrightarrow{u}'$. Also, it must be the case that $d(\overrightarrow{u},\overrightarrow{u}') = 1$, where $d(u,v)$ is Manhattan distance on the triangular grid, and  $s1,s2,s1',s2' \in S \setminus \{ \mathsf{empty} \}$. If we fix $\overrightarrow{u} \in \mathcal{D}$, then there are two $\overrightarrow{u}' \in \mathcal{D}$ that satisfy $d(\overrightarrow{u},\overrightarrow{u}') = 1$.  A movement rule is applicable if it can be applied both (i) locally and (ii) globally, as follows.

(i) Locally, the pair of monomers should be in state $s1, s2$, share bond $b$ and have orientation~$\overrightarrow{u}$ of $s2$ relative to $s1$. Then, one of the two monomers is  chosen nondeterministically to be the {\em base} (that remains stationary), the other is the {\em arm} (that moves).  If the~$s2$ monomer, denoted $X$, is chosen as the arm then $X$ moves from its current position~$p(X)$ to a new position $p(X) - \overrightarrow{u} + \overrightarrow{u}'$. After this movement 
$\overrightarrow{u}'$ is the relative position of the~$s2'$ monomer to the~$s1'$ monomer, as illustrated in Figure~\ref{fig:model}(b). Analogously, if the $s1$ monomer, $Y$, is chosen as the arm then~$Y$ moves from  $p(Y)$ to  $p(Y) + \overrightarrow{u} - \overrightarrow{u}'$. Again,~$\overrightarrow{u}'$ is the relative position of the~$s2'$ monomer to the~$s1'$ monomer.  Bonds and states may change during the movement. 

(ii) Globally, the movement rule may push and/or pull other monomers, or if it can not then it is not applicable. 
This is formalised as  follows, and an example is shown in Figure~\ref{fig:nondeterministicmovement}. Let $\overrightarrow{v} \in \mathcal{D}$ be a unit vector. The $\overrightarrow{v}$-boundary of a set of monomers $Q$ is defined to be the set of grid points outside  $Q$  that are unit distance in the $\overrightarrow{v}$ direction from monomers in $Q$. 
Let~$C$ be a configuration containing adjacent monomers $A$ and $B$, and let~$\config'$ be $\config$ except that the bond between~$A$ and~$B$ is null in $C'$ if not null in $C$. The \emph{movable set} $M = \mathcal{M}(C, A,B,\overrightarrow{v})$  is the smallest subset of $\config'$ that contains~$A$ but not~$B$ and can be translated by~$\overrightarrow{v}$ to give the  set  $M_{+\overrightarrow{v}}$ where the new configuration $C'' = ( C' \setminus  M ) \cup M_{+\overrightarrow{v}} $ is such that:
(a) monomer pairs in $\config'$ that are joined by rigid bonds have the same relative position in $C'$ and $C''$, 
(b) monomer pairs in $\config'$ that are joined by flexible bonds are neighbours in $C''$, 
and (c) the $\overrightarrow{v}$-boundary of $M$ contains no monomers. 
If there is no such set, then we define $ M= \mathcal{M}(C, A,B,\overrightarrow{v}) = \{ \}$.  
If $\mathcal{M}(\config, A, B, \overrightarrow{v}) \neq \{ \}$, then the movement where $A$ is the arm (which should be translated  by~$\overrightarrow{v}$) and $B$ is the base (which should not be translated) is applied as follows: (1) the movable set $\mathcal{M}(\config, A, B, \overrightarrow{v})$ moves unit distance along $\overrightarrow{v}$;  (2) the states of, and the bond between, $A$ and $B$ are updated according to the rule;  (3) the states of all the  monomers besides $A$ and $B$ remain unchanged and pairwise bonds remain intact (although monomer positions and flexible/null bond orientations may change). 
If $\mathcal{M}(\config, A, B, \overrightarrow{v}) = \{ \}$, the movement rule is inapplicable  (the rule is ``blocked'' and in particular $A$ is prevented from translating).

A \emph{nubot system} $\mathcal{N} = (C_0, \mathcal{R})$ is a pair  where $C_0$ is the initial configuration, and $\mathcal{R}$ is the set of  rules. If configuration $C_i$  transitions to $C_j$ by some  rule $r \in \mathcal{R}$ %
we write $C_i \vdash  C_j$. A {\em trajectory} is a finite sequence of configurations $C_1, C_2, \ldots , C_\ell$ where  $C_i \vdash  C_{i+1}$ and $1 \leq i \leq \ell-1$.  A nubot system is said to {\em assemble} a target configuration $C_t$ if, starting from the initial configuration~$C_0$, every trajectory evolves to a  translation of $C_t$. 

A nubot system evolves as a continuous time Markov process. The rate for each rule application %
is 1. If there are $k$ applicable transitions for a configuration~$C_i$ (i.e.\ $k$ is the sum of the number of rule and agitation  steps that can be applied to all monomers), then the probability of any given transition being applied is $1/k$, and the time until the next transition is applied is  an exponential random variable with rate $k$ (i.e.\ the expected time is $1/k$).  The probability of a trajectory is then the product of the probabilities of each of the transitions along the trajectory, and the expected time of a trajectory is the sum of the expected times of each transition in the trajectory. Thus, $\sum_{t \in \mathcal{T}} \mathrm{Pr}[t] \cdot   \mathrm{time}(t)$ is the expected time for the system to evolve from configuration~$C_i$ to configuration~$C_j$, where~$\mathcal{T}$ is the set of all trajectories from~$C_i$ to any  translation of $C_j$, %
and $ \mathrm{time}(t)$ is the expected time for trajectory $t$.

The complexity measure {\em number of monomers} is the maximum number of monomers that appears in any configuration. The {\em number of states} is the total number of distinct monomer states that appear in the rule set. {\em Space} is the maximum area, over the set of all reachable configurations, of the minimum area $l \times w$ rectangle (on the triangular grid) that, up to translation, contains all monomers in the configuration.

The following lemma is used to analyse some of our constructions and was proven in~\cite{nubots}.
\begin{lem}[\cite{nubots}]\label{lem:chernoff}
In a nubot system, if there are $m$ rule applications $a_1, a_2,\dots, a_m$ that must happen, and (1) the desired configuration is reached as soon as all $m$ rule applications happen, 
(2) for any specific rule application $a_i$ among those $m$ rule applications, there exist at most $k$ rule applications $r_1, r_2, \dots, r_k$ such that $a_i=r_k$ and for all $j$, $r_j$ can be applied directly after $r_1, r_2, \dots, r_{j-1}$ have been applied, regardless of whether other rule applications have happened or not,
(3) $m\leq c^k$ for some constant $c$, then the expected time to reach the desired configuration is $O(k)$.
\end{lem}

\subsection{Nubots and decision problems}
Let $\mathbb{N} = \{ 0,1,2,\ldots \}$. $\lineseg{y}{}{}$ denotes a finite length {\em line segment} of nubot monomers. Given a binary string $x \in \{ 0,1\}^{\ast}$, written $x = x_0x_1 \ldots x_{k-1}$, we let $\lineseg{\nb{x}}{}{}$ denote a  line segment of $k$ nubot monomers that represent $x$ using one of two ``binary'' monomer states.   $|$\lineseg{\nb{x}}{}{}$| \in \mathbb{N}$ denotes the number of monomers in~$\lineseg{\nb{x}}{}{}$. Given a line of monomers~$A$ composed of~$m$ line~segments, the notation $\lineseg{A}{i}{}$ means segment $i$ of $A$, and  $\lineseg{A}{i}{j}$ means monomer $j$ (or sometimes the bit encoded by monomer $j$) of segment~$i$ of~$A$.  We next define what it means to decide a  language (or problem) using nubots. 

\begin{definition}\label{def:nubotdecide}
A finite set of nubot rules $\mathcal{N}_L$ decides a language $L \subseteq \{ 0,1 \}^*$ if for all $x \in \{ 0,1\}^*$ there is an initial configuration $C_0$ consisting only of the horizontal line $\lineseg{\nb{x}}{}{}$ of monomers, where by applying the rule set $\mathcal{N}_L$, the system always eventually reaches a configuration containing only a single ``answer'' monomer which is in one of two states: (a)  ``$\mathrm{accept}$'' if $x \in L$, or  (b) ``$\mathrm{reject}$'' if $x \not\in L$.  Further, from the time it first appears, the answer monomer never changes its state. 
\end{definition}

\subsection{Boolean circuits and the class $\NC$}\label{sec:circuitsDefs}
  We define a Boolean circuit to be a directed acyclic graph, where the nodes are called gates and each node has a label that is one of: input (with in-degree 0), constant~0 (in-degree 0), constant~1 (in-degree 0), $\vee$ (OR, in-degree $1$ or $2$), $\wedge$ (AND, in-degree $1$ or $2$), $\neg$ (NOT, in-degree $1$).  
One of the gates is  identified as the output gate, which has out-degree 0.  
    The \emph{depth} of a circuit is the length of the longest path from an input gate to the output gate. 
  The \emph{size} of a circuit is the number of gates it contains. 
  Besides the output gate, all other gates have  out-degree  bounded  by the  circuit size.
We work with {\em layered} circuits: gates on layer $i$ feed into gates on layer $i+1$. 
  A circuit computes a Boolean (no/yes) function on a fixed number of Boolean variables,  by the inputs and constants defining the output gate value in the standard way. In order to compute functions over an arbitrary  number of variables, we
  define (usually, infinite) families of circuits.  
    We say that a family of circuits
  ~$\mathcal{C}_L = \{ c_n \mid c_n \textrm{ is a circuit with } n \in \mathbb{N}  \textrm{ input gates} \}$ decides a language~$L \subseteq \{ 0,1\}^*$ if for each~$x\in\{0,1\}^{\ast}$ circuit~$c_{|x|}
  \in \mathcal{C}_L$ on input $x$ outputs~$1$ if~$w \in L$ and~$0$ if~$w \notin
  L$.

  In a \emph{non-uniform} family of circuits there is no required similarity, or relationship, between family members.    
  In order to specify such a requirement   
  we use a \emph{uniformity function} that 
  algorithmically specifies 
  some similarity between members of a circuit family. 
  Roughly speaking, a \emph{uniform circuit family}~$\mathcal{C}$ is an infinite sequence of
  circuits with an associated function~$f : \{1\}^* \rightarrow
  \mathcal{C}$ that generates members of the family and is
  computable within some resource bound.  
   Here we care about logspace-uniform circuit families:

  \begin{definition}[logspace-uniform circuit family]
    A circuit family $\mathcal{C}$ is
    logspace-uniform, if there is a function  $f : \{1\}^* \rightarrow
  \mathcal{C}$ that is  computable on a  deterministic logarithmic space Turing machine, and where $f(1^n) = c_n$ for all $n \in \mathbb{N}$, 
  and $c_n \in  \mathcal{C}$ is a description of
  a circuit with $n$ input gates.
  \end{definition}
Without going into details, we assume reasonable descriptions (encodings) of circuits as strings. We note that there are stricter, but more technical to state, notions of uniformity in the literature, such as AC$^0$ and DLOGTIME uniformity~\cite{AK2010,greenlaw1995limits,murphy2012and}.  We do not require anything less powerful than logspace uniformity here as our main result is a lower bound on nubots' power, hence the more expressive the uniformity condition on circuits, the better (although most of the common circuit classes are reasonably robust under these more restrictive definitions anyway).

Define  $\NC^i$ to be the class of all languages $L \subseteq \{  0,1 \}^*$ that are decided by  $O(\log^i n)$ depth, polynomial size logspace-uniform Boolean circuit families. Define $\NC = \bigcup_{i=0}^{\infty} \NC^i$, in other words $\NC$  is the class of languages decided by polylogarithmic depth and polynomial size logspace-uniform Boolean circuit families. Since $\NC$ circuits are of polynomial size, they can be simulated by polynomial time Turing machines, and so $\NC  \subseteq  \P $.
It remains  open whether this containment is strict~\cite{greenlaw1995limits}.  
See~\cite{vollmer1999introduction-short} for more on circuits. 

The complexity class $\NL$ is the set of languages accepted by nondeterministic Turing machines that have a read-only input tape and a single worktape of length  logarithmic in the input length.

\section{Example: A nubots line doubling routine}\label{sec:linedouble}
\begin{figure}[!h]
  \centering
    \includegraphics[width=0.88\textwidth]{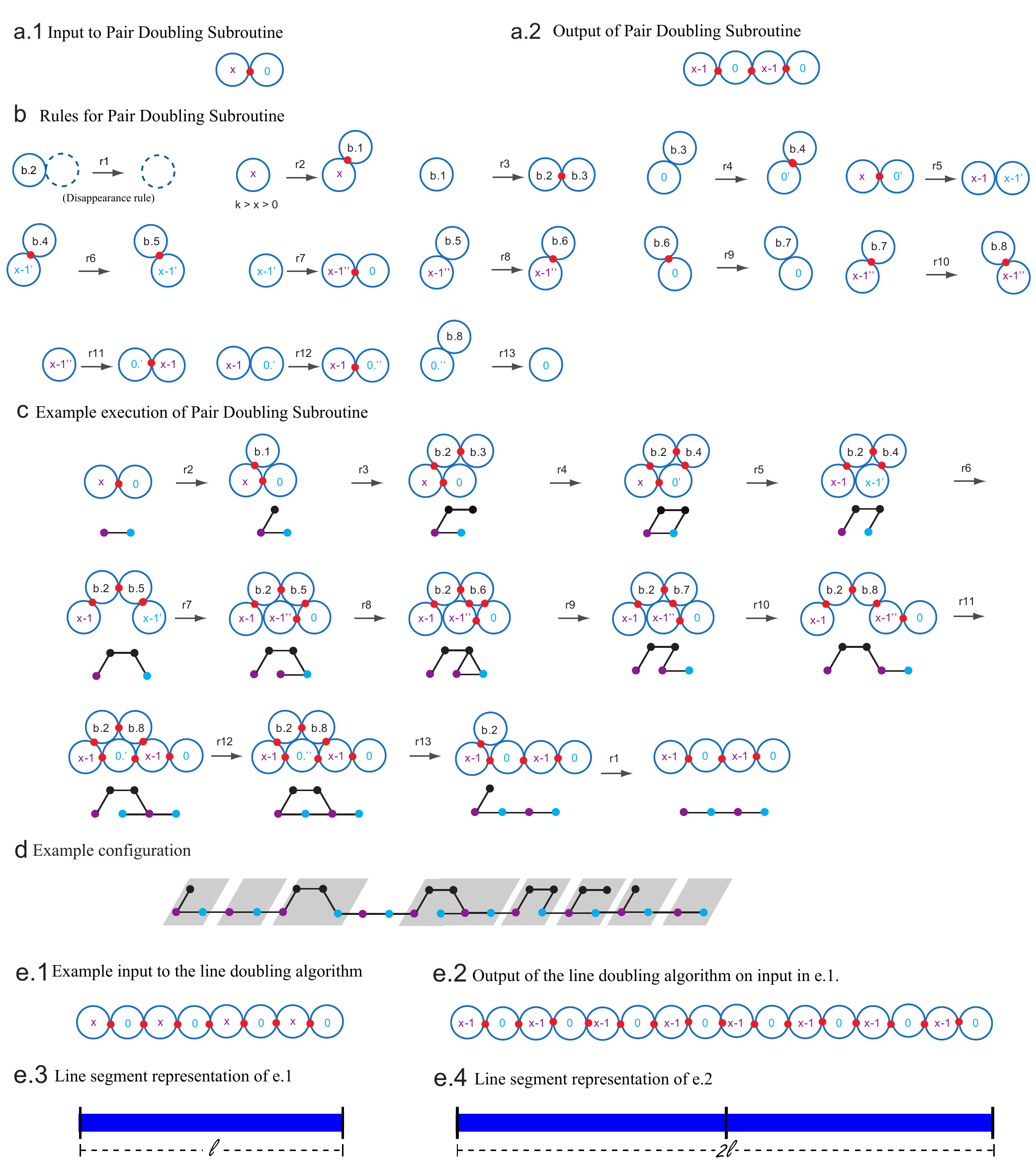}
  \caption{Line doubling nubot algorithm for a line of $l$ monomers that uses a technique from~\cite{nubots}.  (a.1)~Input, (a.2)~output and (b)~rule set for the pair doubling subroutine (PDS). The input and output monomers have alternating blue/purple states (with numbers as shown). Since the LHS of Rule $r_i$ is the RHS of Rule $r_{i-1}$ for $i >2$, the rules must be applied sequentially.  (c)~Example execution of PDS.  (d)~Example configuration of a line undergoing length doubling with concurrent applications of PDS to demonstrate the asynchronous nature of the algorithm. (e.1)~Example input for the line doubling algorithm. (e.2)~Example output for the line double algorithm. (e.3) \& (3.4)~A simplified ``line segment'' representation of (e.1) \& (e.2)~used throughout the paper.}
  \label{fig:lineDoubling}
\end{figure}

This section describes a simple construction with the goal of familiarising the reader with the nubot model. We give an algorithm for doubling the length of a line of $l$ monomers in $O(\log l)$ expected time. This algorithm is essentially a simplification of the line growth algorithm in~\cite{nubots}, and it will be used in later sections of the paper. We first describe the algorithm then provide a proof for correctness and a time and space analysis. 

We require that the input line be comprised of monomers of alternating states, i.e. every monomer in the input line is in one of two unique states with the property that no two adjacent monomers are in the same state. This property of the line is preserved at the end of the line doubling routine.
\begin{lem} A length $l$ line of monomers can be doubled to length $2 l$ in $O(\log l)$ expected time,  $O(1)$ states  $O(l) \times O(1)$ space.
\end{lem}

\begin{proof} {\bf Algorithm description. \ \ }  
The algorithm uses concurrent applications of the pair doubling subroutine (PDS) described in Figure~\ref{fig:lineDoubling}. As described in more detail below, the algorithm treats the input line of $l$ monomers as a line of ${l}/{2}$ monomer pairs that can double in length independently of each other, for even $l$. After the execution of the subroutine, a monomer pair is transformed into two monomer pairs in alternating states different from the original pair. This ensures that each pair of monomers in the input line can only double in length once during the course of the entire algorithm execution. Thus, the length of the input line is doubled by the end of the algorithm, which terminates when every monomer pair in the input has been doubled in length via the subroutine. For odd $l$, the same thing happens for $\lfloor {l}/{2} \rfloor$ monomer pairs, and the rightmost monomer simply adds a single new monomer to its right.

PDS begins with a pair of monomers with states $x, 0$ and ends with four monomers in states $ x-1, 0, x-1, 0$. Figure~\ref{fig:lineDoubling}a provides an example input and output of the line doubling algorithm, where monomers are shown as left (purple), right (blue) pairs. The rules for PDS are given in Figure~\ref{fig:lineDoubling}b and an example execution is shown in Figure~\ref{fig:lineDoubling}c.  Each monomer on the line assumes either the ``left'' or the ``right'' state: left is colored purple, right is colored blue. The initial $x_{\mathrm{left}}, 0_{\mathrm{right}}$ monomers send themselves to state $(x-1)_{\mathrm{left}},0_{\mathrm{right}}$ while inserting two new monomers to give the pattern  $(x-1)_{\mathrm{left}},0_{\mathrm{right}}, (x-1)_{\mathrm{left}}, 0_{\mathrm{right}}$. To achieve this, the initial pair of monomers create a ``bridge'' of 2  monomers on top and, by using movement and appearance rules, two new monomers are inserted. The bridge monomers are then deleted and we are left with four monomers. Throughout the execution, all monomers are connected by rigid bonds so the entire structure is connected. PDS completes in constant expected time 13 as shown in Figure~\ref{fig:lineDoubling}c since there are a total of 13 rules for PDS that must be applied sequentially, as shown in Figure~\ref{fig:lineDoubling}b.

PDS has the following properties: (i) during the application of its rules to an initial pair of monomers $x_{\mathrm{left}}, 0_{\mathrm{right}}$  it does not interact with any monomers outside of this pair, and (ii) a left-right pair creates two adjacent left-right pairs. These properties imply that along a partially formed line, multiple subroutines can  execute asynchronously and in parallel, on disjoint left-right pairs, without interfering with each other.

\paragraph{Correctness.} To demonstrate that the algorithm doubles the length of the line correctly, it is sufficient to demonstrate that the following invariant holds throughout the algorithm execution and that the algorithm terminates. Every left/right pair of monomers in states $x_{\mathrm{left}}0_{\mathrm{right}}$ the input becomes replaced by two left/right monomer pair in states $(x-1)_{\mathrm{left}} 0_{\mathrm{right}}(x-1)_{\mathrm{left}} 0_{\mathrm{right}}$.
Locally, the invariant holds from the fact that PDS takes a pair of left/right monomers in states $x_{\mathrm{left}}, 0_{\mathrm{right}}$ as shown in Figure~\ref{fig:lineDoubling}a.1 and outputs four monomers in states $(x-1)_{\mathrm{left}},0_{\mathrm{right}}, (x-1)_{\mathrm{left}}, 0_{\mathrm{right}}$ as shown in Figure~\ref{fig:lineDoubling}a.2, with Figure~\ref{fig:lineDoubling}c demonstrating that PDS does this correctly. 
Since PDS can be applied to each monomer pair independently of any other pair, adjacent concurrent applications of PDS will not block each other. To see that the algorithm terminates, we note that since the input and the output of PDS assume different states and PDS can only double monomer pairs in the input states, each pair of monomers in the original input line can undergo PDS exactly once. 

\paragraph{Time and space analysis.}
As shown in Figure~\ref{fig:lineDoubling}c, the space complexity of PDS is $4 \times 2$. Since PDS only attaches monomers on top of the input monomers as per the rules, adjacent monomer pairs in the input of the line doubling algorithm will remain on the same axis (i.e. maintain their $y$-coordinates on the triangle grid shown in Figure~\ref{fig:model}a).  Thus, the space complexity of the line doubling algorithm is $O(l) \times O(1)$. We have established above that the expected time for PDS is $13$. The event in which an application of PDS takes place is a Poisson process; therefore, the expected time for a single occurrence of this event to take place is $1/k$, where $k$ is the total possible positions for PDS to be applied. Let $T$ be the time it takes for the line doubling algorithm to terminate on an input of length $l$, then the expected value of $T$ is  
$E[T] = 13   \sum_{i=1}^{l/2} 1/i  = O(\log l) $.
\end{proof}

\section{Using synchronization to communicate quickly}\label{sec:sync}

In previous work~\cite{nubots} a fast signalling method, called synchronization, was introduced for the nubot model. Here, we use the term ``shift  synchronization'' for this technique, and introduce another kind of synchronization called ``lift synchronization''. With these two synchronization mechanisms, we can send one of two distinct messages (bits) to all monomers on  a line  in expected time that is merely logarithmic  in the line length. 

\begin{figure}[t] 
\begin{center}
    \includegraphics[width=0.58\textwidth]{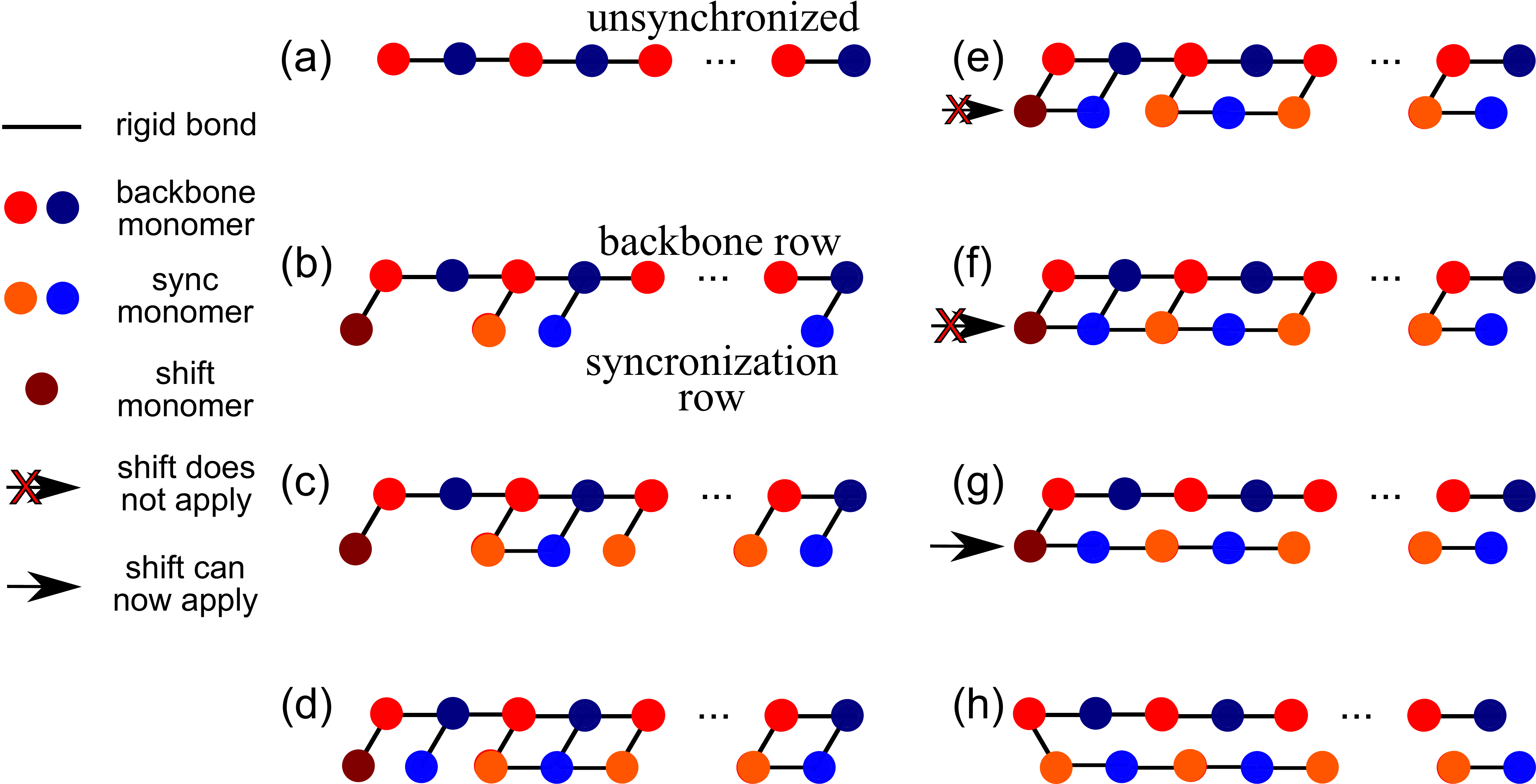}
  \caption{Shift synchronization, see reference~\cite{nubots} for more details. (a)~Initial state. (b)~Monomers, randomly and in parallel, each grow a new {\em synchronization monomer} below. The leftmost new monomer (in brown) is denoted the ``shift monomer''. (c)--(d)~When synchronization monomers detect neighbouring horizontal synchronization monomers to the left and right, they bond. When a synchronization monomer has bonded to both horizontal neighbours, its bond to its parent monomer is removed. (e)~When the shift monomer detects a synchronization monomer neighbour to the right, it changes state, permitting a movement rule to be applied, although the connectivity prevents this movement from occurring yet. (f)~Synchronization monomers continue to appear and update their bond structure. (g)--(h)~All of the vertical rigid bonds are gone and the movement rule can now be applied in one step. All monomers on the original horizontal line  detect the change in state (parity) of their neighbour below.}
  \label{fig:shiftsync}
  \end{center}\end{figure}

\begin{lem}[Communication via synchronization]
Let $\ell$ be a length $n$ line of monomers, where each monomer in $\ell$ is in one of two distinct states $\{ s_0,s_1\}$, with each adjacent pair distinct from each other.  A  bit $b \in \{ 0,1\} $ can be communicated to all monomers on the line in $O(\log n)$ expected time, $O(1)$ monomer states and  $ O( n) \times O(1)$ space. \label{lem:sync communication}
\end{lem}

\begin{proof}
We first give a brief overview of shift synchronization using  Figure \ref{fig:shiftsync},  more details can be found in~\cite{nubots}. Each monomer on the line, in state $s \in \{ s_0,s_1\}$, attaches a new {\em synchronization monomer}  below itself with state $s'$ and with a rigid bond. When a synchronization monomer with state $s'$ senses a new horizontally adjacent neighbouring synchronization monomer it forms a rigid (horizontal) bond with this monomer. After connecting to both neighbouring synchronization monomers, the monomer removes the bond between it and its parent monomer (with state $s$) above.

The rightmost and leftmost synchronization monomers are treated differently. At the rightmost end of the line, the new monomer requires only one bonded neighbour (to the left) before removing its bond to its parent monomer. The leftmost synchronization monomer is called the ``shift monomer''. This shift monomer attempts to push the (new) synchronization row to the right. However, by definition of the movement rule, the shift monomer can move only after  all of vertical rigid bonds between the synchronization row and the original line have been removed. Also, due to the order in which bonds are formed and removed, this can only happen after the entire synchronization row has grown. At some point, we are guaranteed to get to the configuration in Figure~\ref{fig:shiftsync}(g), where the shift monomer is free to push right. After the move (Figure~\ref{fig:shiftsync}(h)), the relative position of synchronization monomers to their generating monomers has changed. Thus, the original line of monomers are free to detect that synchronization has occurred, and a 0 bit has been communicated to all of them. 

\begin{figure} \centering
    \includegraphics[width=0.58\linewidth]{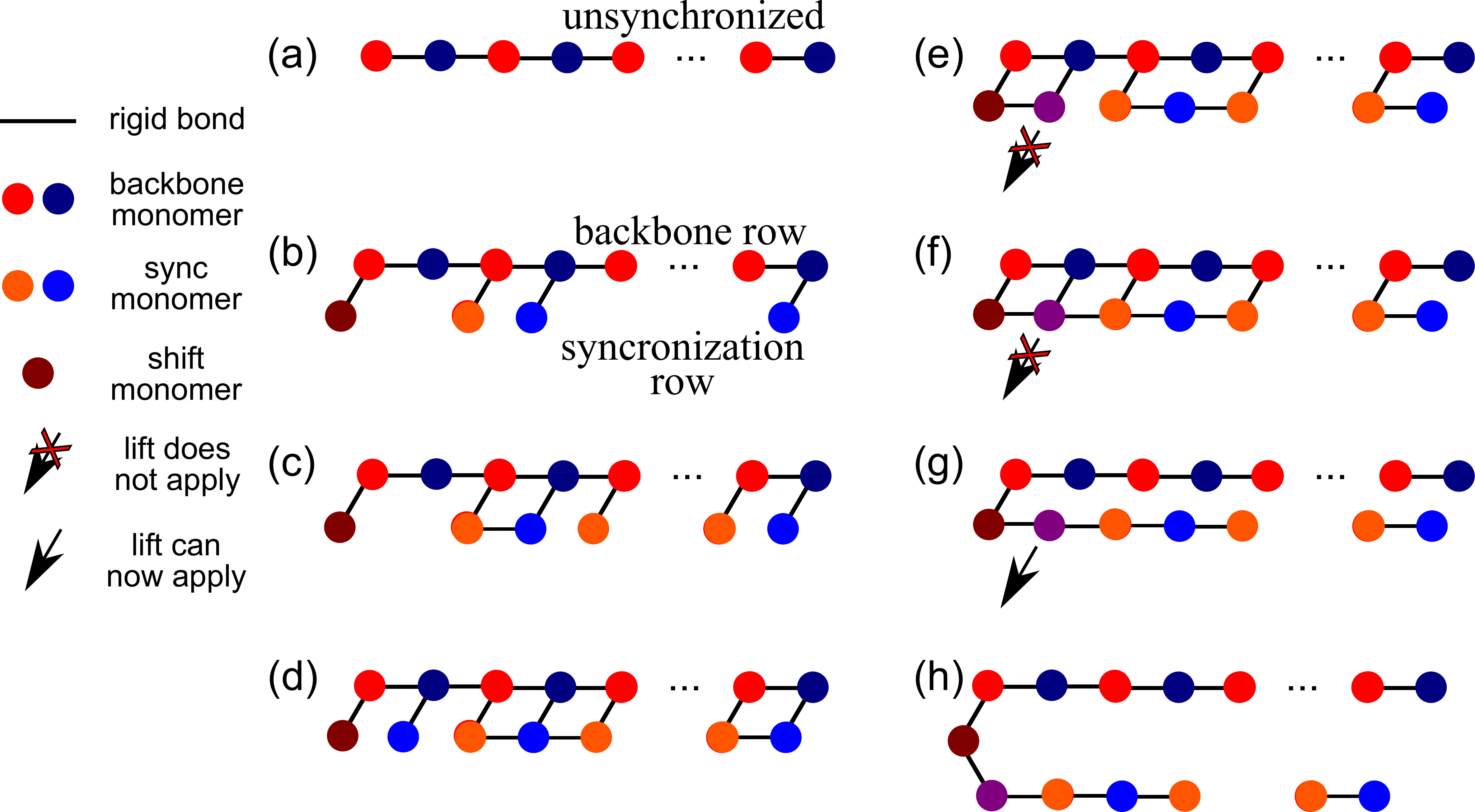}
  \caption{Lift synchronization. Lift synchronization uses similar ideas to shift synchronization, except instead of pushing the entire synchronization row  horizontally, the synchronization row is moved vertically below, and away, from the original line. The monomers on the original line are then free to detect the disappearance of synchronization monomers, signalling the completion of the lift synchronization.}
  \label{fig:liftsync}
\end{figure}

To send a 1 bit we use a similar method, called lift synchronization, shown in Figure~\ref{fig:liftsync}. In lift synchronization the synchronization row is lifted vertically down, and away, from the original line, rather than being shifted right. As with shift synchronization this can only occur after the entire synchronization row has been built and all bonds are in their final form. After the move (Figure~\ref{fig:liftsync}(h)), the monomers on the original line detect the new empty space below,  and thus detect that a 1 bit has been communicated to them. 

In this way, for a line in any of the 6 rotations, it is possible to communicate a 0 or 1 bit, depending on whether shift or lift synchronization is used. The expected time to send the bit is $O(\log n)$, as (a) all new monomers are created independently and in parallel, and (b) each monomer needs only to wait on a constant number of neighbours in order to get its bond structure to the final configuration. The space and states bounds are straightforward to see.  
\end{proof}

\section{Fast line growth using $O(1)$ states}\label{sec:line}
The line growth algorithm given in \cite{nubots} grows a line of length $n\in\mathbb{N}$ in $O(\log n)$ time, using $O(\log n)$ monomer states and starting from a single monomer on the grid. Here, we provide an alternative line growth algorithm that completes in $O(\log^2 n)$ time, using $O(1)$ monomer states and starting from $O(\log n)$  monomers on the grid. Although our construction is an $O(\log n)$ factor slower than that in~\cite{nubots}, it uses only $O(1)$ states while maintaining  the property that all growth is contained within an $O(n) \times O(1)$ region. The latter two properties are both requirements in achieving our main theorem via the other constructions in this paper, which extensively use this line growth algorithm. 

\begin{problem}[Binary Line Growth problem]\label{def:Binary Line Growth problem}   
Input: A line of $\lfloor \log_2 n \rfloor + 1 $ monomers each in one of two binary states from $ \{ s_0 , s_1\}$, that encode the binary string $b = b_{\lfloor \log_2 n \rfloor } \ldots b_1 b_0$ in the standard way, where $n = \sum_{i=0}^{\lfloor \log_2 n \rfloor } b_i \cdot 2^i$ . \\
Output: A line of $n$ monomers. 
\end{problem}

\begin{thm}[Binary Line Growth] \label{thm:binary line growth}
There is a nubot algorithm to solve the Binary Line Growth problem   in expected time $O(\log^2 n)$,  space $O(n) \times O(1)$, and with $O(1)$ states. 
\end{thm} 

\begin{proof}  
As described in the problem statement, the input $n$ is encoded as a line of $\lfloor \log_2 n \rfloor +1$ monomers where the $i$th monomer  encodes bit $b_i$ of the binary  string $b = b_{\lfloor \log_2 n \rfloor } \ldots b_1 b_0$, and where~$b$  encodes $n \in \mathbb{N}$ in the usual way.  The construction proceeds iteratively: at iteration~$ k $, where $0 \leq k \leq \lfloor \log_2 n \rfloor$, bit~$b_k$ is read from the input and if $b_k = 1$ the partially grown line is increased in length by the value $2^k$, otherwise the length of the line remains unchanged. The idea is described at a high-level in the algorithm in Figure~\ref{fig:algo-seq}, below we show  that the  integer variables in that algorithm can be implemented as lines of the corresponding integer lengths, and these can be acted upon in a way that quickly builds the length~$n$ line. 

\paragraph{Construction details.}
During construction, the line-growing configuration is composed of three main regions. The first is the ``input'', as described above; at  iteration $k$ of the algorithm the least significant bit $b_k$ of the input is read (stored), and deleted. Then we have a working region containing two lines, respectively called the ``generator'' and the ``mask'',  each of which have length $2^k $ at iteration~$k$. Finally we have the ``line'' under construction: at iteration~$k$, the line length is given by the binary number $b_{k-1} \ldots  b_1 b_0$ encoded by the first $k$ bits (LSBs) of the input.

\newcommand{\mask}{\ensuremath{\mathrm{mask}}}
\newcommand{\gen}{\ensuremath{\mathrm{generator}}}
\newcommand{\dwline}{\ensuremath{\mathrm{line}}}
\begin{figure}[t]
\centering
\begin{minipage}[c][][c]{0.87\textwidth} 
\begin{algorithm}[H]
{\footnotesize 
\SetAlgoLined \DontPrintSemicolon \SetSideCommentRight

line\_growth($b = b_{\lfloor \log_2 n \rfloor}  \ldots  b_2 b_1 b_0$ : binary string that encodes $n \in \mathbb{N}$) \;
Initialise:  $k = 0$,  $\mask = 1$,  $\gen = 1$,  $\dwline = 0$ \;
\While{$k \leq \lfloor \log_2 n \rfloor$  }{    
   $\mask := 2 \cdot \mask$ \hfill\tcc{$\mask := 2\cdot 2^k = 2^{k+1}$}
  
  \If{$b_k = 0$ }{ 
      $\gen := 2 \cdot \gen$  \hfill\tcc{$\gen := 2\cdot 2^k = 2^{k+1}$}
  
  }
  \ElseIf{$b_k = 1$ }{
        $\gen := 3 \cdot \gen$  \hfill\tcc{$\gen := 3\cdot 2^k$}
  
       $\dwline := \dwline + \gen - \mask $  \hfill\tcc{$\dwline := \dwline + 2^k$}
       
       $\gen :=  \gen - (\gen - \mask)$   \hfill\tcc{$\gen := 3\cdot 2^k - 2^k =  2^{k+1} $}

  }
  Delete $b_k$  
  
  ++$k$
}
\Return $\dwline$
} 
\end{algorithm}
\end{minipage} 
\caption{Algorithm that takes a binary string $b$ as input, that encodes $n \in \mathbb{N}$, and returns the integer $n$. This algorithm describes the control flow for the nutbots construction that builds a line of length $n$ in the proof of Theorem~\ref{thm:binary line growth}.}
\label{fig:algo-seq}
\end{figure}

\begin{figure}[h!]
    \centering 
    \includegraphics[width=0.95\textwidth]{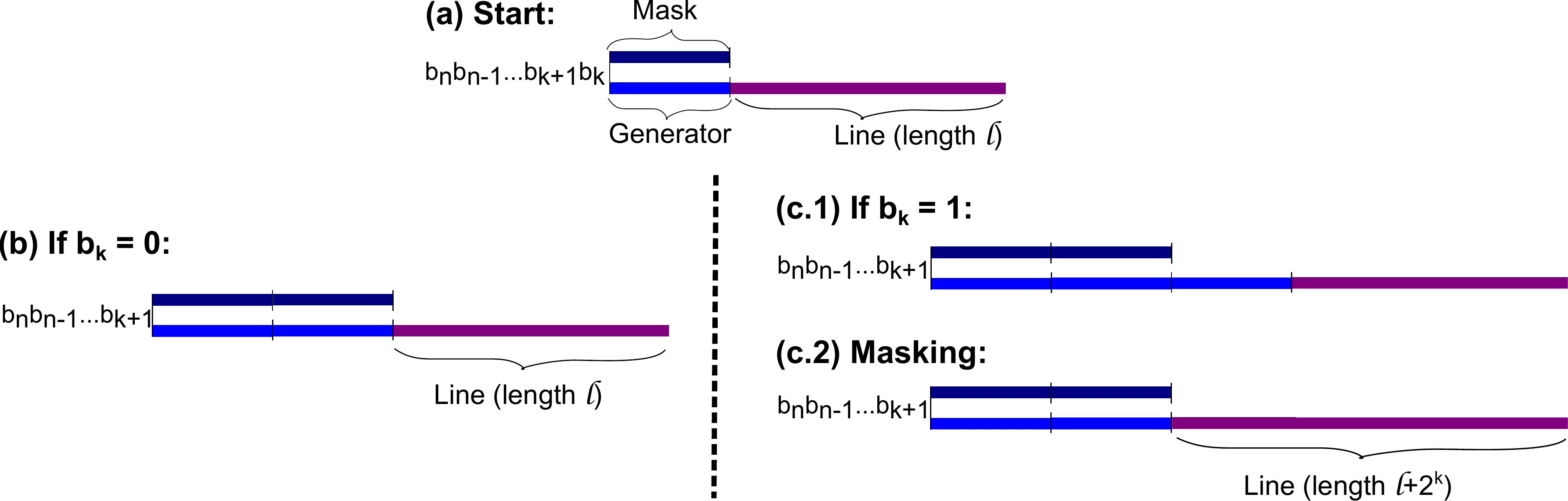}
  \caption{Reading a single input bit, and growing the  line accordingly. (a)  From the start state, depending on if the least significant bit remaining in the string (bit~$b_k$, the $k^{th}$ bit of the original string) is a 0 or 1, the system will end up in one of two different configurations, shown in~(b) or~(c.2). More details for the $b_k=1$ case are shown in Figure~\ref{fig:tripling}.}
  \label{fig:iteritiveoverview}
\end{figure}

The construction begins with  the rightmost of the input monomers growing a small, constant-size, hardcoded structure containing both the generator and mask, both initialised to be of length $1$. 

Figure~\ref{fig:algo-seq} describes a (seemingly  overcomplicated, but analogous to our construction) algorithm for generating the integer $n$ from a bit string $b$. Our construction implements this algorithm, but where the integer variables ``\mask'', ``\gen'', ``\dwline'' are encoded in unary as lines of monomers of that length. It is straightforward to verify, via induction on~$k$, that upon input of the string $b \in \{ 0,1\}^{\ast} $, that encodes  $n \in \mathbb{N}$, the algorithm in Figure~\ref{fig:algo-seq} returns the integer $n$.  Our nubot implementation of one iteration of this algorithm  is shown in Figure~\ref{fig:iteritiveoverview}. Figure~\ref{fig:iteritiveoverview} uses a high-level notation where lines of nubot monomers are represented as colored lines  drawn on the square grid.  We describe the construction    by  describing the main primitives  it uses to  implement the algorithm in Figure~\ref{fig:iteritiveoverview}: line doubling or tripling implement multiplying by 2 or 3; synchronization implements bit communication---and thus which instructions to implement next---to all monomers; and masking implements taking differences.

\begin{figure}
\centering
    \includegraphics[width=0.9\textwidth]{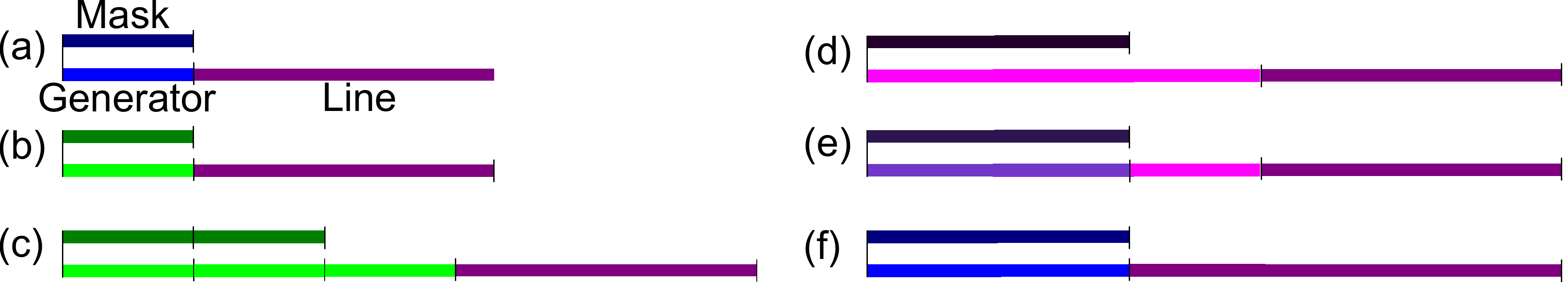}
    \vspace{1ex}
  \caption{Reading a 1 bit. (a) Initial configuration. (b) Synchronization message sent to the mask (dark blue) and generator (light blue)  lines to initiate  tripling of the generator. (c) Masks doubles in length, generator triples. (d) Synchronization. (e) Masking:  monomers in the generator look immediately above for a corresponding monomer in the mask line. If none exists, the generator monomer changes its state to that of a line monomer, if one exists it stays part of the generator. (f)~Masking finished, synchronization. }
  \label{fig:tripling}
\end{figure}

\paragraph{Line doubling and tripling.} Line doubling takes a line of length $\ell$ and generates a line of length~$2\ell$, as described in Section~\ref{sec:linedouble}. Line tripling takes a line of length $\ell$ and generates a line of length $3\ell$, using a similar technique  (rather than inserting~2 monomers, we insert~1, synchronize, then insert~1 again), hence  we omit the details. 

\paragraph{Synchronization \& communicating a bit.}  We use a synchronization algorithm to simultaneously switch a line of monomers into a single shared state.  As described in Section~\ref{sec:sync}, we have the two methods of lift and shift synchronization: we use one  to communicate a 0 bit and the other to communicate a 1 bit to monomers in the generator~and  mask.

\paragraph{Masking.} For two lines of different lengths, $\ell_1 > \ell_2$,  {\em masking}  communicates their difference $\ell_1- \ell_2$ to the line of greater length~$\ell_1$.  The lines are assumed to be orientated parallel, touching, and horizontal with their leftmost extent at the same~$x$ position. Assume the shorter line is on top: it synchronizes (by growing a new synchronization row on top), then the longer line synchronizes (by growing a new  synchronization row on bottom). Then the monomers in the longer line detect the presence or absence of monomers on the shorter line  above: if there is a monomer above then the longer line monomer goes to state $s_1$, if not it goes to state $s_2$. See Figure~\ref{fig:tripling}(d)-(f) for an outline.

\paragraph{Final steps.} The final bit of the input to be read is $b_{\lfloor \log_2 n \rfloor } =1$ (the MSB of a binary number is always 1) and just before reading it the line length is $n-2^{\lfloor \log_2 n \rfloor}$. Upon reading the final bit $b_{\lfloor \log_2 n \rfloor }$ some message passing occurs (via synchronizations) to trigger the deletion of the mask and to cause the generator monomers to change state so that they are now part of the line.  This latter step adds $2^{\lfloor \log_2 n \rfloor}$ (generator length) to the line, giving the desired line length of $n$.

\paragraph{Time, space, and states analysis.}
Line doubling/tripling of a length $n$ line happens in expected time $O(\log n)$, as does synchronization. There are $O(\log n)$ iterations each with a constant number of doublings/triplings and synchronizations, hence the total expected time is $ O(\log^2 n)$.  The three lines (mask, generator, line) are of length $\leq n$ and with their synchronization rows the height needed is 4, giving a space bound of  $O(n) \times O(1)$. A straightforward analysis of the algorithm shows that~$O(1)$ states are sufficient.
\end{proof}

\newcommand{\seglength}{\ensuremath{(\lfloor \log_2 n\rfloor + 1)}}
\newcommand{\seglengthnoparens}{\ensuremath{\lfloor \log_2 n\rfloor + 1}}
\newcommand{\seglengthminusone}{\ensuremath{\lfloor \log_2 n\rfloor}}

\section{Fast parallel sorting}\label{sec:sorting}

In this section we show how, on the nubot model, to sort $n$ binary numbers, taken from the set $\{ 0,1,\ldots , n -1 \}  $, in polylogarithmic expected  time and a constant number of states. 
Our sorting algorithm is loosely inspired by the work of Murphy et al.~\cite{MurphySort06} who show that physical techniques can be used to sort numbers that are represented as  the magnitude of some physical quantity.  They show that a variety of physical mechanisms can be thought of as an implementation of fast parallel sorting, including gel electrophoresis and chromatography (molecular weight), rainbow sort~\cite{schultes2006rainbow} (frequency), and mass spectrometry (mass to charge ratio). However, our  construction needs to take care of the fact that ours is a robotic-style geometric model that needs to implement fast growth while handling blocking and other geometric constraints.  A similar algorithm works for variations on this problem, such as  sorting $< n$ such numbers, but we omit the details of that.

We first define the distinct element sorting problem and then formally state the result.   

\begin{problem}[Distinct element sorting problem]\label{def:sorting}
Input:  A line of monomers, denoted $A = \lineseg{A}{1}{}\lineseg{A}{2}{}\ldots \lineseg{A}{n}{}$,  composed of $n \in \mathbb{N}$ contiguous line segments where for each $i \in \{1,2,\ldots n \}$  line segment $\lineseg{A}{i}{}$ is of length $|\lineseg{A}{i}{}| = \seglengthnoparens$ 
and encodes a distinct binary number from $\{0,1,\ldots,n-1 \}$, where specifically, for all $j \in \{ 1,2,\ldots  \seglengthminusone  \}$ it is the case that monomer  $\lineseg{A}{i}{j}$ is in one of two binary  states from $ \{ s_0, s_1\} $ and the end-of-segment monomer $\lineseg{A}{i}{\seglengthnoparens}$ is in one of two binary end-of-segment states  $  \{ s\#_{0}, s\#_{1}\}$.    \\
Output:  A line $A'$ consisting of the  $n$ binary line segments sorted in increasing order of  the standard lexicographical ordering of their binary sequences. 
\end{problem}

\begin{thm}[Distinct element sorting]\label{thm:Distinct Element Sorting}
Any instance $A$ of the distinct element sorting problem is solvable on the nubot model in expected time $O(\log^3 n )$, space $O(n \log n) \times O(n)$, and  $O(1)$ monomer states.  
\end{thm}

\begin{figure}
\begin{center}
    \includegraphics[width=0.9\textwidth]{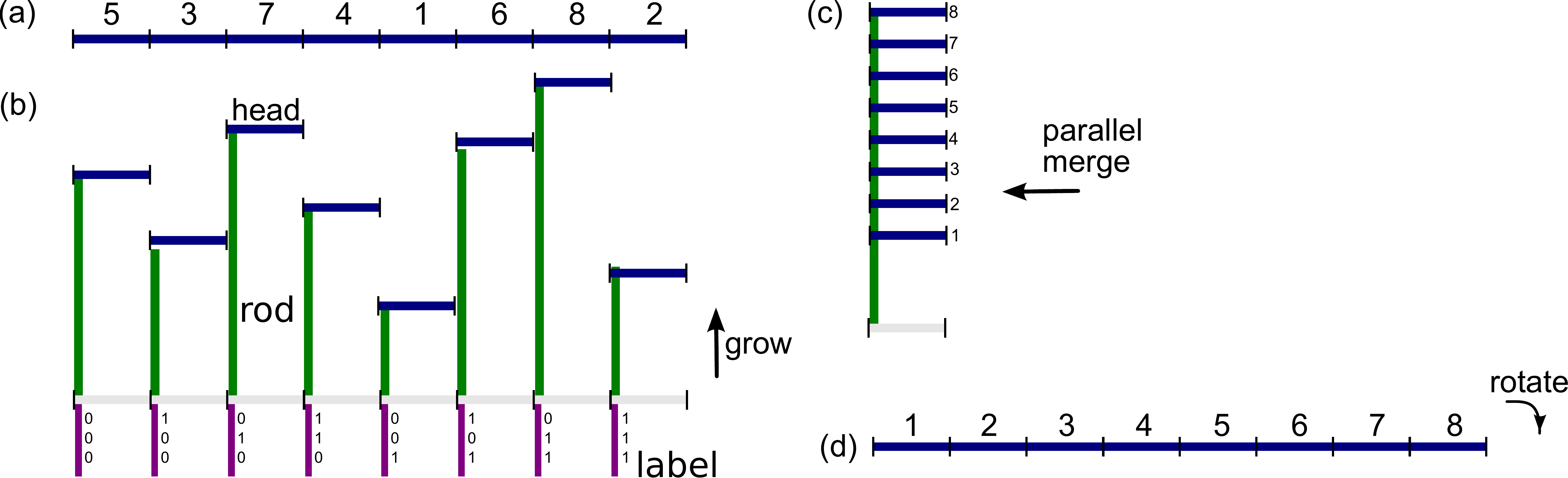}
\end{center}
    \caption{High-level overview of the sorting algorithm.   (a) A line of $n\seglength $ monomers, with~$n$ blue line segments (``heads'') each is the binary representation of a natural number $i \leq n$. (b)  A blue head that encodes value $i$ is grown to height $i$ by a green rod, in time polylog in $i$. 
  Purple ``labels''  are also grown at the bottom. (c) The heads are horizontally merged, using the labels to synchronize, to be vertically aligned. (d) Merged heads rotate down into a line configuration, giving the sorted list. Each stage occurs in expected time polylogarithmic in $n$, more details appear in subsequent figures. }
  \label{fig:sortoverview}
\end{figure}

\begin{proof}
The general idea is as follows. For each element $i$ (encoded as a ``head'') to be sorted, we grow a line of monomers (a ``rod'') to length $i$  as shown  in Figure~\ref{fig:sortoverview}(b).  After doing so, the relative heights of the heads gives their order. We then move each head horizontally left, through a sequence of $O(\log n)$  parallel merging steps, so that all heads are vertically aligned  (Figure~\ref{fig:sortoverview}(c)). Finally,  the heads are rotated and translated so that they lay along a vertical line as shown in Figure~\ref{fig:sortoverview}(d), in increasing order. The details are described next.

\subsection{Sorting details: rod growth and labeling} We begin with an instance of the distinct element sorting problem, an example of which is shown in Figure~\ref{fig:sortbone}(a).

\begin{figure}[t] %
  \begin{center}
    \includegraphics[width=\textwidth]{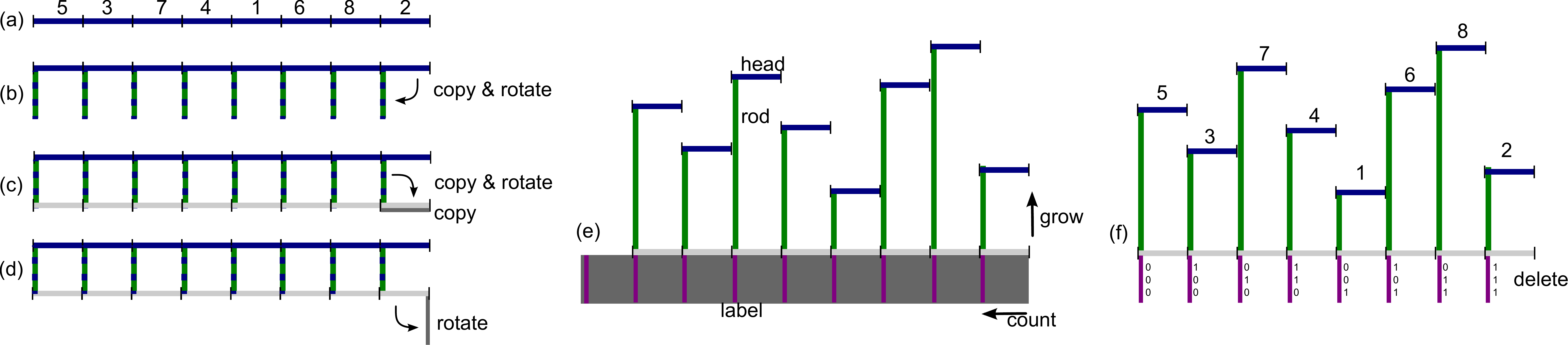}
  \end{center}
  \caption{The beginning of our sorting algorithm, this gives the details for the overview in Figure~\ref{fig:sortoverview}(a)--\ref{fig:sortoverview}(b).  (a) Initial configuration with head monomers in blue. (b) Head monomers are copied and rotated down to vertical (dashed blue-green), and then (c) are copied and rotated down to horizontal to form the light-grey label region. The light-grey region synchronizes after being copied. On the right, the dark grey region is copied and in it rotates down to vertical, as shown in~(d).  In (d) the heads have received a  ``synchronization done'' message from the  light-grey region and (e) they grow a vertical green rod (line) of length equal to the value encoded in the head (note the heads are not connected to each other, except through the green rods and light-grey region).  Also in (e), the dark grey region grows a purple counter, from right to left, that counts from $2^{ \seglengthnoparens} -1$ down to 0 (see main text for details) and is padded with  $ \seglengthnoparens  $ monomers between each counter column (thus growing a~$\seglength  \times \seglength 2^{\seglengthnoparens }$ rectangle,  in dark grey and purple). 
  (f)~All parts of the dark grey region delete themselves, except those directly below a green rod (purple). The purple regions that remain encode $n$ distinct binary strings.}
  \label{fig:sortbone} 
\end{figure}

\paragraph{Initialization.} The monomers begin in binary states as described in Definition~\ref{def:sorting}.  Growth begins at each of the $n$ blue heads: the head is copied and rotated down to vertical as shown in Figure~\ref{fig:sortbone}(b). This  rotation of a $O(\log n)$ length line takes $O(\log\log n)$ expected time to complete using the parallel ``arm rotation'' method in~\cite{nubots}---that is, each monomer independently rotates by one position, relative to its leftmost neighbour.  After rotation (Figure~\ref{fig:sortbone}(b)), each blue-green  line independently synchronizes, then makes a copy of itself which is in turn rotated down to become one of the $n$ horizontal light-grey line segments  shown in Figure~\ref{fig:sortbone}(c). After all light-grey segments are horizontal, they bond to each other and synchronize. This entire process completes in expected time $O(\log n)$, using Lemma~\ref{lem:chernoff}, and is dominated by the synchronization process. 

\paragraph{Grow rods.}  After this synchronization step, as shown  in Figure~\ref{fig:sortbone}(c), the  rightmost grey line segment is copied to form a dark grey segment that is copied down to vertical  in Figure~\ref{fig:sortbone}(d). Also in Figure~\ref{fig:sortbone}(d), and triggered by the previous synchronization, the blue-green rods, in $O(\log n)$ expected time, signal the heads to disconnect from each other, and the blue-grew rods then begin ``growing upwards''.  This vertical growth of the rods implements  a form of counting: we want the rods to grow to the height encoded by their blue head. This is carried out by using the line-growth algorithm in Section~\ref{sec:line} which takes time $O(\log^2 n)$ (an alternative method would be to use a suitable counter, such as the one described below). After  a rod has grown to the value encoded in its  head, shown in Figure~\ref{fig:sortbone}(e), the rod synchronizes, this latter step taking expected time $O(\log n)$. After expected time $O(\log^3 n)$ all $n$ rods have synchronized.

\paragraph{Label growth.} Rod growth occurs  above the light-grey line. Below that line another process takes place, the purpose of which is to label each rod with its position (from right), as a binary number in purple. Here the dark-grey line (Figure~\ref{fig:sortbone}(d), on right)  grows a ``padded'' counter, from right to left. The result of this counter is shown in Figure~\ref{fig:sortbone}(e) and is a $\seglength  \times \seglength 2^{\seglengthnoparens}$  rectangle where each of the purple columns, from right to left, encodes a distinct value from $2^{ \seglengthnoparens} $ down to 1, with the grey regions in between being there for padding purposes only. 

This counter works as follows. The counter is a modified version of the one used in Section 6.2 of~\cite{nubots}; the counter  in~\cite{nubots} used $O(\log n)$ states, here we use $O(1)$ states.  First note that the dark grey strip is of height $\seglengthnoparens $, it begins  counter growth by  converting each of its monomers to a state that represents the bit~1, giving the binary representation of the number $2^{\seglengthnoparens} -1 $.   
Let $j = \seglengthnoparens $, and we begin from the single dark grey column, applying the following procedure iteratively to each new column until  $j=1$. Each column copies itself to the left and in the new column the  $j$th bit is flipped. Both columns then decrement their value of $j$, and {\em both} iterate the copy and bit-flip procedure.   As is the case in~\cite{nubots}, this process  happens asynchronously and independently to all columns. After this happens we have a  $\seglength \times 2^{ \seglengthnoparens}$ rectangle containing all of the purple columns.  We are not done yet: we wish for the purple counter columns to align themselves with the $n$ green rods which are distance $ \seglengthnoparens $ apart, as in Figure~\ref{fig:sortbone}(e). To achieve this, another round of column insertion (i.e.\ counting)  begins, so that between each pair of counter columns, exactly $k = \seglengthminusone  $ new columns are inserted (between each pair of purple columns we  are implementing a  counter that  counts from $k$ down to $1$; note that the integer $k$ is available since the purple counter columns are of height $k$). Now the purple counter rows are exactly  distance $k+1 = \seglengthnoparens$ apart.  When the  process is complete the bottom row of the entire rectangle synchronizes, to give the structure illustrated in  Figure~\ref{fig:sortbone}(e) (although the grey rectangle extends further to the left than shown).

To give a straightforward time analysis, we assume that the copying and decrementing for an individual column happens sequentially and so takes expected time $O(\log n)$. Then in the completed counter, each counter column is the result of  no more than  $O(\log n)$ column-copying operations, hence any monomer in the final rectangle depends on  the application of $O(\log^2 n)$ rules. Applying Lemma~\ref{lem:chernoff} gives an expected time of  $O(\log^2 n)$.  The final synchronization step costs $O(\log n)$ expected time, giving a total   expected time of $O(\log^2 n)$. 

\paragraph{Deletion and synchronization.} All columns of the dark grey region then delete themselves, {\em except} those that are directly below a green rod. The deletion events happen in time $O(\log n)$.  The purple regions that remain are $n$ counter columns that encode $n$ distinct binary numbers. After  each green rod (above) has synchronized it signals to the light-grey line. As each counter row below completes deletion, it too  signals to the light-grey line. The light-grey line undergoes a lift synchronization.  The system is now in the configuration shown in Figure~\ref{fig:sortbone}(f). 

\paragraph{Analysis.} As already discussed, rod growth and the subsequent synchronization of all $n$ rods takes expected time $O(\log^3 n)$, and label growth takes expected time $O(\log^2 n)$.

\begin{figure}[t]
  \begin{center}
    \includegraphics[width = 0.99\textwidth]{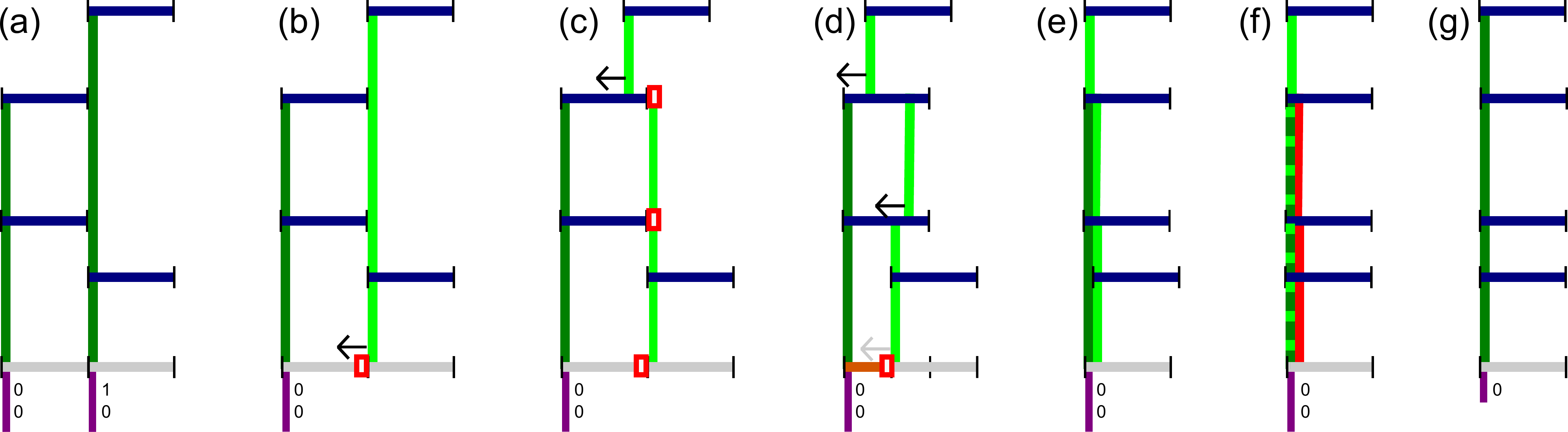}
  \end{center}
  \caption{Merging the heads on two adjacent rods, this gives the details for the overview in Figure~\ref{fig:sortoverview}(c). Exactly one round of parallel (pairwise) head-merging has already occurred, and so each green rod has two blue heads.  (a) Two rods, with two heads each: the goal is to merge all 4 heads onto the left rod. (b) Due to how they were generated the LSBs (top bits)  on the two purple labels are distinct. If the LSB bit is 1, the rod moves left, and deletes the purple label monomers. The rod tries to move left by having the light-grey line sequentially delete its $O(\log n)$ monomers one at a time, although here the rigid rod is immediately blocked due to collisions with the heads on the left.  (c)~Collisions are marked in red. The rod monomers at the collision locations delete themselves. The new shorter rods  can continue moving to the left, by ``walking'' along the blue heads as shown in~(c) and~(d).  (e)~The rods on the right make contact with the rods on the left. (f)~The contact triggers a ``done'' state to be reached by the rod on the left. It also signals for the rod on the right to delete itself. Head monomers from the right are shifted to their new rod. (g)~When everything has moved into place, synchronization occurs along the single green rod, and the LSB (top) bit of the purple label is deleted.}
  \label{fig:mergepart}
\end{figure}

\subsection{Sorting details: merging}
Now that  all  rods have grown, and are labeled,  we will now merge them as shown in in Figure~\ref{fig:sortoverview}(c).

\paragraph{Main idea.} Intuitively, we would like to simply shift all of the heads  to the left, deleting any rods that get in the way.  However, if we are not careful, rods can block each other and significantly slow down the process so that it no longer runs in time polylogarithmic in $n$ (consider the worst case, where the shortest rod is the rightmost one, and we wish to move all heads to the left).  Our merging algorithm gets around this issue by  merging  in a pairwise fashion. Every second pair of heads merge, deleting one of the rods and then the  light-grey line synchronizes. We are left with $\lceil n/2 \rceil$ rods, each having two heads. Then every second pair of those merge, and so on for $O(\log n)$ iterations. To organise the correct order of mergings, we use the purple labels, specifically their binary sequences, which are shown in Figure~\ref{fig:sortbone}(f). 

\paragraph{Merging algorithm.}  The following procedure is iterated until there is exactly one rod left. For each $i$, the $i$th rod checks its label, if the least significant bit (LSB) of the purple label is 1 (in Figure~\ref{fig:sortbone}(f) the LSB is on top), then  rod $i$  attempts  to merge with rod $i-1$ to its left, by ``moving'' its head to the left  distance $\seglengthnoparens $. This pairwise merging process is  described in the caption of Figure~\ref{fig:mergepart}. Rod~$i$, and its label (but not its head) get deleted in the process. After merging of the pair of rods is complete, rod $i-1$ deletes its LSB, thus shortening its purple region by 1. Rod $i-1$ now has two heads, and signals to the light-grey line that it is done. After all pairs have merged, we have $\lceil n/2 \rceil$ rods, with 2 heads each. At this point the light-grey line synchronizes and the process iterates. After $O(\log n)$ rounds of pairwise merging we are left with one rod, which has no label and is  carrying all~$n$ heads.  

When a pair of rods are merging, the right rod needs to move distance $ \seglengthnoparens $ to the left. In the worst case there are $\lceil  n/2 \rceil$ collisions for a pair of rods, however, these are all resolved in parallel as described in Figure~\ref{fig:sortbone}(c). So we have $\leq \lceil  n/2 \rceil$ heads, that each need to independently walk distance $\lfloor \log n \rfloor +1 $ to the left which na\"ively takes  $O(\log^2 n )$ expected time, and by applying Lemma~\ref{lem:chernoff} takes~$O(\log n )$ expected time. 

\paragraph{Final steps.} After merging is complete, the heads are on a single rod, sorted vertically upwards in increasing order of head value. The heads rearrange themselves on the rod to that they are separated by vertical distance exactly $\seglengthnoparens$, and then  rotate down into a line configuration, giving the sorted list as shown in Figure~\ref{fig:sortoverview}(d). 

\subsection{Sorting details: time, space and states analysis}
The expected time to complete the various stages of sorting was given above, and is dominated by growing and synchronizing the rods, which is $O(\log^3 n)$.   For the space analysis, note that the length of the light-grey line is $O(n \log n)$ (giving the horizontal space bound). The rods are of height $O(n)$, and the purple labels are of height $O(\log n)$, giving a vertical space bound  of $O(n)$. Hence we get a  space bound of $O(n \log n) \times O(n)$. All counters and line growth algorithms use number of states that is constant, which can be seen by a careful analysis of each part of the construction.  
\end{proof}

\section {Fast Boolean matrix multiplication}\label{sec:matrixmult}

Let $M$ and $N$ be $n \times n$ Boolean matrices. Let $M_{i, j}$ denote the element at row $i$ and column $j$ of $M$, and let $MN$ denote their matrix product. The following two definitions are illustrated in Figure~\ref{matrixAndNubots} and describe our encoding of  a square matrix as an addressed line of monomers.\footnote{Our choice of a 1D, rather than 2D, encoding simplifies our constructions. It would also be possible use a more direct 2D square encoding, which, it turns out, can be unfolded to and from our line encoding in expected time $O(\log n)$. We omit the details.}  

\begin{figure}
  \begin{center}
    \includegraphics[width=0.5\textwidth]{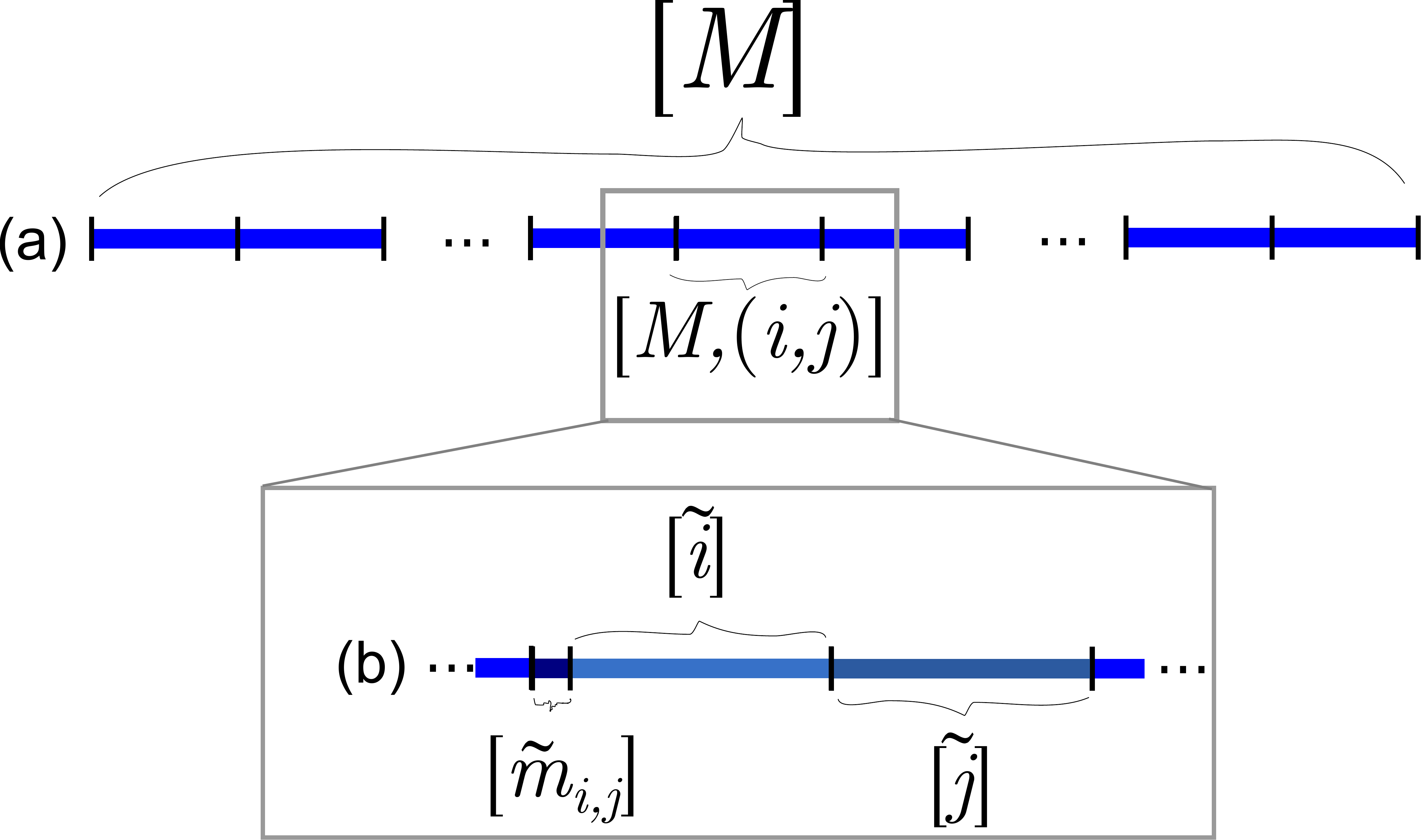}
    \caption{(a) Encoding of a Boolean matrix $M$ as a line of monomers $\lineseg{M}{}{}$. (b) Zoom-in of the encoding of a single matrix entry $M_{ij} \in \{ 0,1\}$ as a line segment $\lineseg{M}{(i,j)}{}$ that contains a single monomer $\lineseg{\nb{m}_{i,j}}{}{}$ that encodes the bit~$M_{i,j}$ and line segments of binary monomers, $[\nb{i}]$ and $[\nb{j}]$, that encode $i$ and $j$.}
    \label{matrixAndNubots}
  \end{center}
\end{figure}

\begin{definition}[Matrix element encoding]\label{def:nubotmatrixelement}
An element $M_{i,j}$ of $n \times n$ Boolean matrix $M$ is encoded in nubot monomers as a line of $O(\log n)$ monomers $\lineseg{M}{(i,j)}{} = \lineseg{\nb{m}_{i,j}}{}{}\lineseg{\nb{i}}{}{}\lineseg{\nb{j}}{}{}$  where $\lineseg{\nb{m}_{i,j}}{}{}$ is a nubot monomer that encodes the bit $M_{i,j}$,  and $\lineseg{\nb{i}}{}{}$ and $\lineseg{\nb{j}}{}{}$ are lines of binary monomers of length $O(\log n)$ that encode the numerical values $i$ and $j$, respectively (the  segments $\lineseg{\nb{i}}{}{}$ and~$\lineseg{\nb{j}}{}{}$ are each terminated by delimiter monomers).
\end{definition}

\begin{definition}[Monomer encoded Boolean matrix]\label{def:nubotmatrix}
An $n \times n$ Boolean matrix $M$ is encoded in nubot monomers as a line of $O(n^2 \log n)$ monomers  $\lineseg{M}{}{}  = \lineseg{M{}}{(1,1)}{} \lineseg{M}{(1,2)}{} \ldots \lineseg{M}{(n,n-1)}{} \lineseg{M}{(n,n)}{}$  of all $\lineseg{M}{(i,j)}{}$ for $1 \leq i,j \leq n$, ordered from left to right, first by $i$, then by $j$.
\end{definition}

The main result of this section, Theorem~\ref{thm:boolmatrixmult}, is a fast parallel algorithm for Boolean matrix multiplication.

\begin{problem}[Monomer encoded Boolean matrix multiplication problem]\label{def:matrixmult}
Input: Monomer encoded Boolean matrices $\lineseg{A}{}{}$ and $\lineseg{B}{}{}$, that represent $n \times n$ Boolean matrices~$A,B$. \\
Output: Monomer encoding of the Boolean  matrix  $\lineseg{C}{}{} = \lineseg{AB}{}{}$. 
\end{problem}

\begin{thm}\label{thm:boolmatrixmult}
The monomer encoded Boolean matrix multiplication problem can be solved in $O(\log^3 n)$ expected time, $O(n^4 \log n) \times O(n^2 \log n)$ space and with $O(1)$ monomer states.
\end{thm}

\begin{figure}[th!]
\begin{center}
    \includegraphics[width=\textwidth]{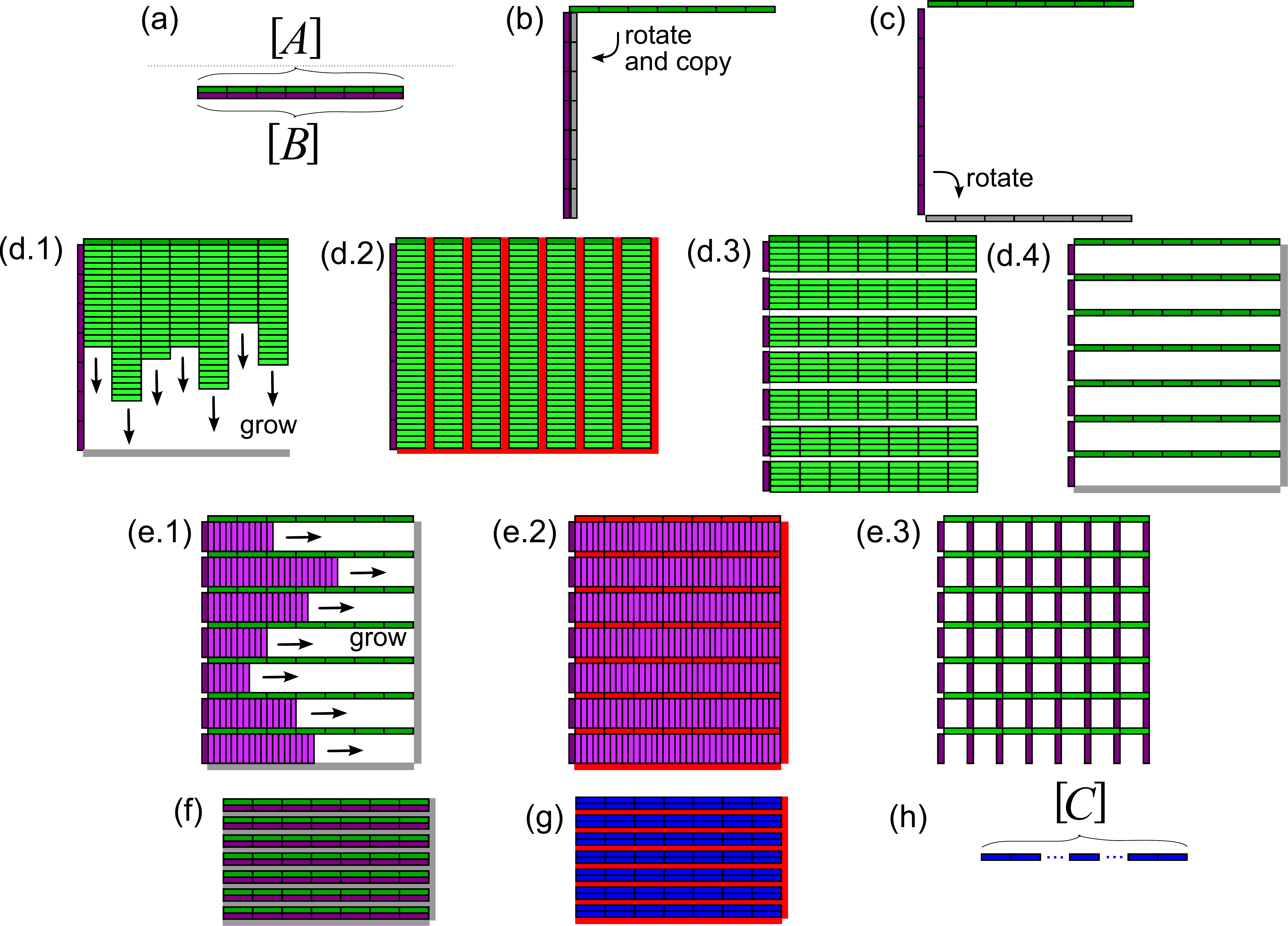} 
\end{center}
  \caption{Parallel function evaluation in 2D, used  in the proof of Lemma~\ref{lem:2dfunction}. (a)~Initial configuration with line $\protect\lineseg{A}{}{}$ in green and line $\protect\lineseg{B}{}{}$ in purple, each has $n$ line segments. We wish to compute $\mathcal{F}$ on all $n^2$ pairs of line segments in $\protect\lineseg{A}{}{}$ and $\protect\lineseg{B}{}{}$. (b)~$\protect\lineseg{B}{}{}$ rotates down to vertical and duplicates.  (c)~The duplicate  of  $\protect\lineseg{B}{}{}$ rotates down to horizontal creating a grey border.  (d.1)~Each segment of $\protect\lineseg{A}{}{}$ duplicates, and the resulting pair of segments duplicate, and so on iteratively. (d.2)~The copied line segments of~$\protect\lineseg{A}{}{}$ reach the bottom grey border line. A vertical gap is inserted between each {\em column} of green line segments, then  synchronization occurs (red). (d.3)~The vertical synchronization causes the system to change connectivity (to be a comb with horizontal teeth),   allowing for segments of~$\protect\lineseg{B}{}{}$ to insert 1-monomer vertical gaps between themselves. (d.4)~Duplicates of $\protect\lineseg{A}{}{}$, not adjacent to a gap delete themselves; monomers rearrange and horizontal synchronization rows are regrown. (e.1)~Segments of $\protect\lineseg{B}{}{}$ duplicate, iteratively. (e.2)~When duplication finishes, synchronizations occur along the copied segments of~$\protect\lineseg{A}{}{}$. (e.3)~Duplicates of segments of $\protect\lineseg{B}{}{}$ not adjacent to the left/right ends of duplicates of segments of $\protect\lineseg{A}{}{}$ delete themselves. (f)~Purple duplicated line segments of $\protect\lineseg{B}{}{}$ rotate up to align parallel with those of $\protect\lineseg{A}{}{}$, the structure shrinks vertically, and a new vertical synchronization row (grey) is formed on the right. (g)~$\mathcal{F}$~is evaluated in parallel on all line segments $\lineseg{A}{i}{}$ and $\lineseg{B}{j}{}$, to give the set of all line segments $ \mathcal{F}(\lineseg{A}{i}{}, \lineseg{B}{j}{} ) $ for all $1 \leq i,j  \leq n $ represented in blue. (h)~The rectangle rearranges into the  line~$\protect\lineseg{C}{}{}$ of length $O(n^2 \log n)$, as  in Figure~\ref{fig:unfolding}.}
  \label{fig:gridmessage}
\end{figure}

\subsection{Parallel function evaluation in 2D}
Before proving Theorem~\ref{thm:boolmatrixmult} we give a useful lemma that formalises a notion of nubots efficiently computing many ($n^2$ here) functions in parallel, where each function acts on two length $k$ inputs.
 Figure~\ref{fig:gridmessage} illustrates the proof. 

\begin{lem}[Parallel function evaluation in 2D]\label{lem:2dfunction}
Let $\mathcal F$ be any function that maps a pair of length $k$ adjacent parallel horizontal monomer lines $\lineseg{X}{}{}, \lineseg{Y}{}{}$  to a length $k$  horizontal monomer line $[Z]$, that is $\mathcal F(\lineseg{X}{}{}, \lineseg{Y}{}{} ) = \lineseg{Z}{}{}$, and moreover~$\mathcal F$ is  nubot computable  in $O(k)$ expected time,  $O(k) \times O(1)$ space, and $O(1)$ states.  Let~$\lineseg{A}{}{} = \lineseg{A}{1}{} \lineseg{A}{2}{} \ldots \lineseg{A}{n}{} $  and $\lineseg{B}{}{} = \lineseg{B}{1}{} \lineseg{B}{2}{} \ldots \lineseg{B}{n}{}$ be monomer lines, each composed of $n$ consecutive length $k$ monomer lines (called ``line segments'').  Then, given $\lineseg{A}{}{}$ and $\lineseg{B}{}{}$ as input,  the line $\lineseg{C}{}{} = \lineseg{C}{1}{} \lineseg{C}{2}{}\ldots \lineseg{C}{n^2}{}$  consisting of all $\lineseg{C}{ i + (j-1)n}{} = \mathcal{F}( \lineseg{A}{i}{}, \lineseg{B}{j}{} ) $ for $1 \leq i,j \leq n$  is computable using nubots  in $O(k + \log^2 n )$ expected time, $O(n^2 k)  \times  O(n k)$  space, and $O(1)$~states.
\end{lem} 
\begin{figure}[t]
\begin{center}
    \includegraphics[width=0.9\textwidth]{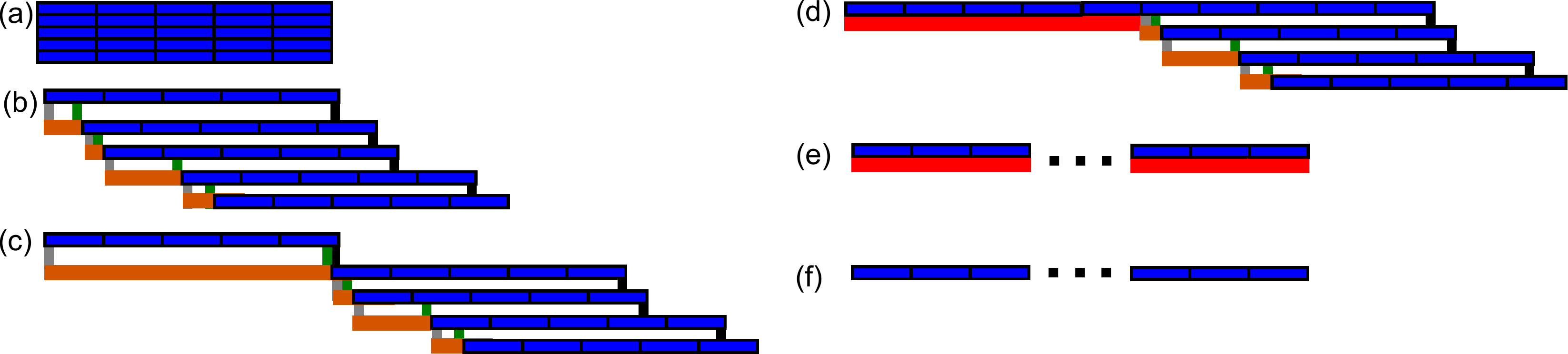} 
    \caption{Unfolding a rectangle of $n^2$ monomer line segments, each of length $O(\log n)$ into a line of length $O(n^2 \log n)$. (a)~Initial configuration. (b) A new ``insertion line'' (shown in orange) grows from the left of each row~$i$. The left end of the orange insertion line at row~$i$ is attached to the left end of  row  $i-1$ (above) by a monomer shown in grey. As the insertion line for row~$i$ grows, it  ``pushes'' row~$i$ to the right, relative to row $i-1$.  A monomer (black) attached to the right end of row $i-1$ and a monomer (green) attached to the right end of the orange insertion line $i$ below are used as ``hooks'' so that the insertion line is stopped from growing beyond length $O(n\log n)$. (c) The green monomer of row 2 and the black monomer of row~1 become \lq\lq{}hooked\rq\rq{}. (d)~Row 2 moves up to be horizontally aligned with row~1. The grey, green, and black monomers delete. When the orange insertion line in the second row is placed adjacent to row~1, it becomes a red  synchronization row. All rows continue  this process independently and in parallel. When all are done the insertion line becomes a synchronization row.  (e)~Ready to synchronize. (f)~Final configuration.}
    \label{fig:unfolding}
\end{center}
\end{figure}

\begin{proof}[Proof of Lemma~\ref{lem:2dfunction}]
Figure~\ref{fig:gridmessage} gives an overview of the construction.  From an initial configuration with~$\lineseg{A}{}{}$ and~$\lineseg{B}{}{}$ adjacent as in Figure~\ref{fig:gridmessage}(a),  $\lineseg{B}{}{}$ rotates down to vertical (Figure~\ref{fig:gridmessage}(b)).  $\lineseg{B}{}{}$ is copied from the grey line which rotates down to horizontal as shown in Figure~\ref{fig:gridmessage}(c). 
In Figure~\ref{fig:gridmessage}(d.1), we duplicate each line segment  $\lineseg{A}{i}{}$, for $1 \leq i \leq n $, $n$ times, down to the grey vertical line, which acts as a barrier to stop the duplication. A one monomer horizontal gap is inserted between adjacent columns of green columns (of line segments), which triggers a vertical synchronization, shown as a vertical red line in Figure~\ref{fig:gridmessage}(d.2) of each  completed green column. Next, monomer-to-monomer messages are passed, horizontally from right to left, within each green  line segment  to signify that monomers  should change from being ``vertically connected'' to being ``horizontally connected''. After this, the vertical red synchronization lines carry out another synchronization and then delete themselves in a way that keeps all green monomers horizontally connected. In Figure~\ref{fig:gridmessage}(d.3), each purple $\lineseg{B}{j}{}$ inserts a 1-monomer vertical gap between between it and its neighbour $\lineseg{B}{j+1}{}$. After all gaps insert, the purple vertical line  synchronizes, and then $n$ horizontal synchronizations happen which tell excess duplicates of $\lineseg{A}{i}{}$ to delete themselves to give  the configuration in Figure~\ref{fig:gridmessage}(d.4).

Next, a duplication and deletion process occurs with $\lineseg{B}{j}{}$ line segments as shown in Figure~\ref{fig:gridmessage}(e) (similar to what we did before, but now horizontally rather than vertically). The $\lineseg{B}{j}{}$'s duplicate until they hit the vertical grey barrier on the right, at which point the system synchronizes. After this occurs, excess $\lineseg{B}{j}{}$ segments are deleted (using direct monomer-to-monomer message transfer as before). When this process is complete, we are at  Figure~\ref{fig:gridmessage}(e.3). 

Next,  the duplicates of each $\lineseg{B}{j}{}$ rotate up to horizontal as shown, and the leftmost copy of~$\lineseg{B}{}{}$ deletes itself in a way that vertically ``shrinks'' the assembly to get Figure~\ref{fig:gridmessage}(f). During this process we make $n$ grey synchronization rows, also shown in Figure~\ref{fig:gridmessage}(f). 
From Figure~\ref{fig:gridmessage}(f) to Figure~\ref{fig:gridmessage}(g), $\mathcal{F} (\lineseg{A}{i}{}, \lineseg{B}{j}{}) $ is computed on each of these  $n^2$ line segments (independently and in parallel), and by the lemma hypotheses this can be done in the allotted space. The horizontal red lines  synchronize, and then the vertical red line synchronizes. After this occurs, we can delete the grey synchronization rows and unfold the result into a line, which is of length $O(n^2 \log n)$, to get the final configuration in  Figure~\ref{fig:gridmessage}(h), using the technique shown in Figure~\ref{fig:unfolding}. 

\paragraph{Space, state and time analysis.}
By stepping through the construction (and Figure~\ref{fig:gridmessage}), it is straightforward to check that the entire construction is contained within space  $O(n^2 k) \times O(n k)$, and uses~$O(1)$ states.

For the time analysis, we first observe that rotation, and copying, of a length~$\ell$ line can each be done in $O(\log \ell)$ expected time via a straightforward analysis~\cite{nubots}. Steps  (b), and (c) of Figure~\ref{fig:gridmessage} involve  rotations and copying of lines of length $O(n \log n)$: this  completes in expected time $O(\log n)$.  The duplication processes of green and purple segments in Figures~\ref{fig:gridmessage}(d) and~(e)  take $O(\log^2 n)$ expected time. Each application of $\mathcal{F}$ takes expected time $O(k)$, and we apply it independently in parallel $n^2$ times, hence via Lemma~\ref{lem:chernoff},  all complete in merely $O(k)$ expected time. There are a number of other places where $\leq n^2$ independent processes, each with expected time $O(\log n)$, take place (deletions in Figures~\ref{fig:gridmessage}(d.4) and (e.3), and rotations in (f)), and by Lemma~\ref{lem:chernoff}, they take expected time $O(\log n)$. In each of Figures~\ref{fig:gridmessage}(d.2), (e.2) and (g) there are  $n$ lines, each of length $O(n \log n)$ that need to be synchronized. For example, in Figure~\ref{fig:gridmessage}(d.2), synchronization for each red vertical single line takes expected time $O(\log n)$, and since we must wait until all $n$ vertical lines are synchronized (independently), and only then synchronize the horizontal line, this takes expected time $O(\log^2 n)$. Finally, the rearrangement  in Figures~\ref{fig:gridmessage}(g) to (h) (given in detail in Figure~\ref{fig:unfolding}) takes expected time $O(\log^2 n)$: each insertion line must grow $O(n \log n)$ monomers before a level is moved up. There are $n$ of them that work independently, so Lemma~\ref{lem:chernoff} gives an expected time to finish of~$O(\log^2 n)$. Besides computing $\mathcal{F}$, the slowest parts of the construction run in expected time $O(\log^2 n)$, and there are at most a constant number of these parts, so the entire construction finishes in expected time~$O(k + \log^2 n)$.  This concludes the proof of Lemma~\ref{lem:2dfunction}.
 \end{proof}

\subsection{Proof of Theorem~\ref{thm:boolmatrixmult}: fast Boolean matrix multiplication}

\begin{proof}[Proof of Theorem~\ref{thm:boolmatrixmult}]
The multiplication $C = AB$ of two $n \times n$ Boolean matrices is defined as 
$ C_{i,j} = \bigvee_{k = 1}^{n} \left(A_{i,k} \wedge B_{k, j}\right)$. To calculate $A_{i,k} \wedge B_{k, j}$,  for each $i,j,k$, we begin by defining the function $\mathcal{F}_{\mathrm{AND}}$ which acts on two encoded matrix elements $\lineseg{A}{(i,k_1)}{}$ and $\lineseg{B}{(k_2,j)}{}$ as follows. 
\[
 \mathcal{F}_{\mathrm{AND}} (\lineseg{A}{(i,k_1)}{}, \lineseg{B}{(k_2,j)}{}) =
  \begin{cases}
     \lineseg{\emptyset}{}{}\lineseg{\nb i}{}{}\lineseg{\nb j}{}{}\lineseg{\nb k_1}{}{}\lineseg{\nb k_2}{}{}& k_1 \neq k_2 \\
     \lineseg{\tilde c_{i,j,k_1}}{}{}\lineseg{\nb i}{}{}\lineseg{\nb j}{}{}\lineseg{\nb k_1}{}{}\lineseg{\nb k_2}{}{} & \text{if } k_1 = k_2
  \end{cases}
\]
where $\emptyset$ is a special monomer denoting ``no useful data here'',   $\lineseg{\tilde c_{i,j,k_1}}{}{}$ is the monomer encoding for the (useful) bit $A_{i,k_1} \wedge B_{k_1,j}$ when $k_1=k_2$, and as usual \lineseg{\nb i}{}{}, \lineseg{\nb j}{}{}, \lineseg{\nb k_1}{}{}, \lineseg{\nb k_2}{}{} denote the binary monomer line segments encoding of $i,j,k_1,k_2 \leq n$.

We now apply Lemma~\ref{lem:2dfunction} to $\lineseg{A}{}{}$ and $\lineseg{B}{}{}$ setting $\mathcal F$ to $\mathcal F_{\mathrm{AND}}$.  This gives a line of monomers with $n^4$ segments, each of length $O(\log n)$, and $n^3$ of which encode useful data. The remainder are the $n^4 - n^3$ segments $\ell$ for which $k_1 \neq k_2$. 
The entire line synchronizes and begins the process of deleting the  useless line  segments $\ell$  as follows.  Each $\ell$ encodes an $O(\log n)$ bit number $p$ (as the concatenation of the bit strings for $i,j,k_1,k_2)$. The digits of  $p$ are used to organise the deletion of the segments. If the LSB of $p$ in segment $\ell$ is 1, then $\ell$ deletes itself. Deletion of a segment works as follows:  the rightmost monomer $r$ of $\ell$ walks on top of its $\ell$ segment, walking left, sequentially deleting monomers until all of the $\ell$ segment is deleted. The monomer $r$ is now adjacent to a new segment $\ell' $ (to the left of the former $\ell$) which causes  $\ell' $ to ``delete'' its LSB (sets its monomer state to $\emptyset$). The entire line grows a shift synchronization row and the process iterates. To stop the iteration: if an $r$ monomer completes its walk left and meets a non-$\ell$ segment (i.e. it meets a useful segment) it initiates growth of a lift synchronization row, when all lift synchronization rows form, a lift synchronization occurs signalling that the entire deletion process has finished.  This gives a line of monomer segments of the form $ \lineseg{\tilde c_{i,j,k_1}}{}{}\lineseg{\nb i}{}{}\lineseg{\nb j}{}{}\lineseg{\nb k_1}{}{}\lineseg{\nb k_2}{}{} $  with $k_1 = k_2$. Next  each of the redundant $\lineseg{\nb k_2}{}{}$ segments deletes itself, then the line synchronizes. This gives  a line  of $n^3$ segments  $\lineseg{C}{(i,j,k)}{}$ for $1 \leq i,j,k \leq n $. 

To calculate the elements $C_{i,j} = \bigvee_{k = 1}^{n} C_{i,j,k}$, for $1\leq i,j\leq n$, we begin by sorting the line segments first by~$i$, then by $j$,  then by $k$.  From this sorted line of elements, for each $i,j $ we calculate $\lineseg{C}{(i,j)}{}$, the encoding of the matrix element $C_{i,j}$, in a two-step  process as follows. First, for all~$i,j,k$ where  $k \neq 1$, the $\lineseg{\tilde{i}}{}{}$ and $\lineseg{\tilde{j}}{}{}$ line segments delete themselves from each of the  $\lineseg{C}{(i,j,k)}{}$  line segments, and the entire line synchronizes when done. The  $n$ segments~$\lineseg{\tilde k}{}{}$, for  all $k\leq n$, then rotate perpendicular to their  original orientation, and translate (horizontally ``shrink'') so that the $n$  monomers of the form  $\lineseg{\tilde c_{i,j,k}}{}{}$  lie horizontally adjacent to each other. At this point we have a structure consisting of vertical columns of $\lineseg{\tilde c_{i,j,k}}{}{}$ and $\lineseg{\tilde {k}}{}{}$, ordered horizontally by $i$, then by $j$ then by $k$, for $1 \leq i,j,k \leq n$ (and also with those $\lineseg{\tilde k}{}{}$ that encode~1, which still have their horizontal~$\lineseg{\tilde i}{}{}$ and $\lineseg{\tilde j}{}{}$ segments). The set of monomer bitstrings $\{ \lineseg{\tilde k}{}{} \, | \, k \leq n \} $ are used to organise iterated pairwise ORing. 
For each $\lineseg{\tilde k}{}{}$ with LSB 1 delete the LSB from the~$\lineseg{\tilde k}{}{}$ segment and OR its $\lineseg{\tilde c_{i,j,k}}{}{}$ bit with its neighbour $\lineseg{\tilde c_{i,j,k'}}{}{}$ to the left, with the result bit being stored in the $\lineseg{\tilde c_{i,j,k'}}{}{}$ monomer to the left. Then a shift synchronization occurs. This is iterated $O(\log n)$ times until all $n$ bits have been ORed and finished with a lift synchronization. 
We are left with a line of segments of the form $\lineseg{\tilde C}{(i,j)}{} = \lineseg{\tilde c_{i,j}}{}{}\lineseg{\tilde i}{}{}\lineseg{\tilde j}{}{}$, for $1 \leq i,j \leq n$, ordered first by $i$ and then by $j$. This is exactly the monomer representation of the matrix $C = AB$ that we desire. 

\paragraph{Time, space and state analysis.}
The time of Boolean matrix multiplication is dominated by  Lemma~\ref{lem:2dfunction} and the sorting algorithm. Since the expected time of $\mathcal{F_{\mathrm{AND}}}$ is $O(\log n)$, then $k = O(\log n)$ in the hypothesis of  Lemma~\ref{lem:2dfunction}, giving  $O(\log^2 n$) as the  expected time of the application of  Lemma~\ref{lem:2dfunction}. There are~$O(n^3)$ monomer segments (each of length $O(\log n)$) to be sorted when calculating the ORs, so the expected time for sorting is $O(\log^3 n)$. Hence, the entire matrix multiplication takes expected time~$O(\log ^3 n)$.

The most space-consuming aspects of Boolean matrix multiplication are  Lemma~\ref{lem:2dfunction} and the sorting algorithm. Lemma~\ref{lem:2dfunction} takes space $ O(m^2 k) \times O(m k)$ for two lines each containing  $m$ line segments each of which is of length $k$. For matrix multiplication, we are starting from two encoded $n \times n$ matrices, each of which has $n^2$ elements. Setting $m = n^2$ and $k = O(\log n)$, the total space for Boolean matrix multiplication is thus   $O(n^4\log n) \times O(n^2\log n)$, and since the   sorting algorithm takes less space than that, we are done. 

A careful analysis of the algorithm shows that the  number of monomer states is $O(1)$. 
\end{proof}

\section{Boolean circuit simulation}\label{sec:circuitsim}

Our main result, Theorem~\ref{thm:nubotsSolveNC} is a restatement of the following theorem.  
Definition~\ref{def:nubotdecide} defines what it means to decide  languages with the nubot model.

\begin{thm}
\label{thm:circuitsim} 
Let  $L \in \mathrm{NC}^j$ be a  language decided by a logspace-uniform Boolean circuit family $\mathcal{C}$ of circuits that have depth $O(\log^j n )$, for some $j \geq 1$,  size $O(n^k)$ and input length $|x|=n$, and let~$\lineseg{\nb{x}}{}{}$ denote the representation of  $x \in \{ 0,1\}^*$ as a line of binary monomers. 
Then there is a set of nubot rules $\mathcal{N}_L$ such that, 
for all $x \in \{ 0,1\}^*$, starting from an initial configuration containing only  $ \lineseg{\nb{x}}{}{}$, $\mathcal{N}_L$ decides whether $x \in L$ and uses space $n^{O(1)} \times n^{O(1)}$,  monomer states $O(1)$, and expected time  $O(\log^{j+3} n)$. \end{thm}

The proof is contained in Section~\ref{sec:CirGenSim}, where we give a nubot algorithm that given $x$, quickly generates a Boolean circuit, and then simulates that circuit on input $x$.  Before that, in Section~\ref{sec:TMsim}, we present a nubot algorithm that simulates function-computing logspace  Turing machines in polylogarithmic expected time. This fast Turing machine simulation will be used in the circuit generation part of  Section~\ref{sec:CirGenSim}.

\subsection{Fast parallel simulation of space bounded Turing machines}\label{sec:TMsim}
Here, we give a polylogarithmic expected time simulation of deterministic logspace Turing machines that compute functions with domain and range $\{ 0,1\}^*$.

\begin{lem}
\label{lem:LTMtonubots}
Let $\mathcal{M}$ be a deterministic Turing machine that  on input  $x \in \{ 0,1\}^*$, of length $|x| = n$, generates an output $y \in \{ 0,1\}^*$, in $O(\log n)$ workspace and time~$t$.   There is a set of  nubot rules~$\mathcal{N}_{\mathcal{M}}$ such that for all $x \in \{ 0,1\}^*$, starting with the  initial configuration containing only the line~$\lineseg{\nb{x}}{}{}$ (that represents~$x$), $\mathcal{N}_{\mathcal{M}}$ computes~$\lineseg{\nb{y}}{}{}$  (the representation of  $y$) using  $O(1)$ states, 
$n^{O(1)} \times n^{O(1)}$ 
 space, and $O(\log^4  n)$ expected time. 
\end{lem}

Before giving the  proof of Lemma~\ref{lem:LTMtonubots} we state some assumptions  about  $\mathcal{M}$:
\begin{enumerate}
  \item $\mathcal{M}$ follows the standard conventions for logspace Turing Machines: there are 3 tapes: a read-only input tape, a $O(\log n)$ space bounded  work tape, and a write-only output tape whose  head  moves in one direction only. 
  \item $\mathcal{M}$ uses the alphabet $\{ 0,1 \}$ on all 3 tapes (the input is delimited with the symbol $\#$). 
  \item A configuration consists of the input tape head position and read symbol, worktape contents and head position, worktape read symbol, and machine state, and (unusually\footnote{Our configurations include an output tape write symbol and an output tape head position which is not standard practice~\cite{Papadimitriou94}, but will be useful in our construction.}) the output tape head position and write symbol. 
  \item  $\mathcal{M}$ always ends its computation in a halting, accept state. 
\end{enumerate}

    There are $n$ possible positions for the head on the input tape and $ O(\log n)$ head positions  on the work tape. We note that each configuration of $\mathcal{M}$ can be written as a string over $\{0,1 \}^*$ of length $O(\log  n)$. This follows from the fact that in a given configuration $O(\log n)$ bits describe the position of the input and output tape heads, $O(\log \log n)$ bits describe the position of the worktape head,
    $O(\log n)$ bits describe the contents of the work tape, and $O(1)$ symbols describe the read and write symbols on the various tapes,    and the machine state. Thus,  on length-$n$ input, $\mathcal{M}$  visits at most  $2^{O(\log n)}  = n^{O(1)}$  configurations before halting, or looping forever, in other words   $t = n^{O(1)}$.
We next define the \emph{configuration matrix} of a space-bounded Turing machine. 

\begin{definition}[Configuration matrix]\label{def:configmatrix}
Let $\mathcal{M}$ be a deterministic  Turing machine with space bound~$s(n)$. Consider the set  of all $k = n^{O(1)}$ possible configurations on a length $n$ input (we include all syntactically valid configurations for a worktape with $s(n)$ tape cells, even though on a given input many will be unreachable). We define  $M$ to be the $k \times k$ Boolean matrix where for $1\leq i, j \leq k$,  $m_{i,j}=1$  if and only if there exists a one-step transition from configuration $c_i$ to configuration $c_j$, and $m_{i,j}=0$ otherwise.
\end{definition} 

\begin{proof}[Proof of Lemma~\ref{lem:LTMtonubots}]
A logarithmic space-bounded Turing machine $\mathcal{M}$ can be simulated efficiently  in parallel (in polynomial time, using polynomial processors/resources) in a variety of parallel models by  iterated squaring of~$\mathcal{M}$'s Boolean configuration matrix $M$~\cite{Papadimitriou94}. Specifically, we can determine whether there exists a sequence of one-step transitions from any configuration~$c_i$ to any $c_j$ by beginning with matrix $M$, computing  $M := M^2 + M$, and iterating this procedure $O(\log n)$ times. A path between the two configurations exists only if entry $m\rq{}_{i,j} =1$ in the resulting matrix~$M\rq{}$. We call  matrix $M'$ the \emph{path-complete matrix}.
Since $\mathcal{M}$ is deterministic and always accepts ($\mathcal{M}$ is total: for any input $1^n, n\in \mathbb{N}$ it outputs a circuit $c_n$), there is exactly one  path, through $\mathcal{M}$'s configuration graph, that leads from the start configuration to the halt (accept) configuration. The technique of iterated squaring is sufficient for simulating  Turing machines that decide languages, but here we want to simulate a function-computing  machine.  We  do this by  modifying the iterated squaring technique: our configurations contain (extra)  information about what is written to the output tape, we  appropriately  extract this information during our simulation of $\mathcal{M}$ by iterated squaring. The remainder of the proof describes how we do all of this in the nubot model.

We generate all possible configurations of Turing machine $\mathcal{M}$, in parallel. First, we build a counter that counts up to $2^{s(n)} = n^{O(1)}$, the upper bound on the number of distinct worktape contents  of~$\mathcal{M}$.
Once completed, each row of the counter then generates its own counter, counting up to the number of different positions that the head can be on the input tape. This process is iterated for each of the (constant number of)  attributes in a Turing machine configuration to give a final counter with $k = n^{O(1)}$  rows, one for each distinct configuration (see Definition~\ref{def:configmatrix} for~$k$, and see~\cite{nubots} for details on efficiently growing a counter). The counter backbone synchronizes, giving a $k \times O(\log k)$ rectangle, whose bond structure forms a  ``comb''.   The counter rearranges itself into a line  $\lineseg{C}{}{} = \lineseg{C}{1}{}\lineseg{C}{2}{} \dots \lineseg{C}{k}{}$ where line segment $\lineseg{C}{i}{}$ encodes configuration~$c_i$.

We next use our encoding  $\lineseg{C}{}{}$ of all possible configurations    to generate an encoding of the configuration matrix~$M$. First, $\lineseg{C}{}{}$ is copied so that we have two parallel instances of~$\lineseg{C}{}{}$, side-by-side.  Next, we apply Lemma~\ref{lem:2dfunction}, setting  $\mathcal{F} = \mathcal F_{\mathcal{M}}$ where $\mathcal F_{\mathcal{M}}$  takes as input the pair of parallel line segments $\lineseg{C}{i}{}$ and $\lineseg{C}{j}{}$, and a copy of the input line segment\footnote{Each configuration is of length polynomial in  $|x| = O(|$\lineseg{\nb{x}}{}{}$|)$, hence including $\lineseg{\nb{x}}{}{}$ here does not change the asymptotics.} $\lineseg{\nb{x}}{}{}$ that encodes  $x$, and returns a segment $\lineseg{M}{(i,j)}{} = \lineseg{m}{i,j}{}\lineseg{C}{i}{}\lineseg{C}{j}{}$ where $\lineseg{m}{i,j}{}$ is a binary nubot monomer representing element $m_{i,j}$ in $\mathcal{M}$'s configuration matrix\footnote{The line segments in an encoded matrix usually encode the matrix element's $(i, j)$ coordinates, here we do things slightly differently: we are using encoded configurations, rather than natural numbers, as the matrix indices. This simplifies our constructions a little.}. In other words, given the encoding of two configurations $c_i, c_j$, the function $\mathcal F_{\mathcal{M}}$ determines if there is a one-step transition from~$c_i$ to~$c_j$ via Turing machine~$\mathcal{M}$. $\mathcal F_{\mathcal{M}}$ works by straightforward message-passing and state changes from monomer to monomer along the pair of encoded configurations. To satisfy  the hypotheses of Lemma~\ref{lem:2dfunction}, $\mathcal F_{\mathcal{M}}$ should work in time linear in the encoded configurations' length ($|\lineseg{C}{i}{}|$, $|\lineseg{C}{j}{}|$) which is easily  achieved. It is also the case  that the space and states bound in the hypotheses of Lemma~\ref{lem:2dfunction} are met by $\mathcal F_{\mathcal{M}}$. After applying Lemma~\ref{lem:2dfunction} we get an encoding of the configuration matrix $M$ as a single line  $\lineseg {M}{}{}$  of consecutive line segments 
$\lineseg{M}{(i,j)}{} = \lineseg{m}{i,j}{}\lineseg{C}{i}{}\lineseg{C}{j}{}$ for $1 \leq i,j \leq k$.

We make a copy of the encoded matrix $\lineseg {M}{}{}$ (a line of monomers) and then use Theorem~\ref{thm:boolmatrixmult} to square $\lineseg {M}{}{}$, giving~$\lineseg{M^2}{}{}$ (another line of monomers). After we have both lines $\lineseg{M}{}{}$ and $\lineseg{M^2}{}{}$, \lq\lq{}adding\rq\rq{} the lines together (to compute $M^2 + M$) is easy: matrix elements with the same $(i, j)$ coordinates are adjacent when $\lineseg{M}{}{}$ and $\lineseg{M^2}{}{}$ are orientated parallel and next to each other, so the addition can be carried out ``locally''. The iterated squaring (and addition) are carried out  $1 + \log k  = O(\log n)$ times. The result is an encoding of $M$'s path-complete  matrix $M'$. 

Consider the path-complete configuration matrix~$M'$, with start configuration $c_{\mathrm{start}}$  and halt configuration $c_{\mathrm{halt}}$. 
We need to (i) determine which configurations are on the unique path, in the configuration graph, from $c_{\mathrm{start}}$ to $c_{\mathrm{halt}}$, and (ii) follow this unique path keeping track of what was written to the output tape at each step. 
For any $i$, if configuration $c_{\mathrm{start}}$ leads to configuration $c_i$  in $\geq1$ steps then $M'_{c_{\mathrm{start}} , c_i} = 1$.  Similarly, if configuration $c_i$ leads to configuration $c_{\mathrm{halt}}$   in one or more steps then $M'_{c_i, c_{\mathrm{halt}}} = 1$. Hence it is sufficient to extract row $c_{\mathrm{start}}$ and column $c_{\mathrm{halt}}$ from $M'$, and compare them, in order to find the entire path of configurations from $c_{\mathrm{start}}$ to  $c_{\mathrm{halt}}$. 

We do this by first deleting all encoded matrix elements of $\lineseg{M'}{}{}$ that are not in row $c_{\mathrm{start}}$ or not in column $c_{\mathrm{halt}}$. This results in two lines of monomers, that are then aligned parallel and side-by-side. Next each entry $i$ is compared, if there is a 1 in both we keep the entry, otherwise the entry is deleted. We are left with the list of configurations on  the path from $c_{\mathrm{start}}$ to $c_{\mathrm{halt}}$. This line of monomers synchronizes.

Next, the encoded matrix entries (bits) are deleted leaving the list of encoded configurations. Configurations that do not write anything to the output tape are deleted.  The remaining configurations are sorted in increasing order of output-tape write location. Since the output tape head moves one way only, this gives the list of outputting configurations in the order they are executed by the Turing machine $\mathcal{M}$.  Finally, all monomers that do not represent a symbol written by the output tape head are deleted, and the result is compressed into a line. We are left with a monomer line encoding $y$, the output tape contents. 

\paragraph{State, space and time analysis of Turing machine simulation.}
The state complexity of $O(1)$ can be seen from stepping through the algorithm.
In the proof, the configuration matrix dimension size is $k\times k $, where $k = n^{O(1)}$.  From Theorem~\ref{thm:boolmatrixmult}, matrix multiplication for two $k \times k$ matrices takes space  $O(k^4 \log k) \times O(k^2 \log k)$ using nubots, and since $k = n^{O(1)}$ this gives the space bound in the lemma statement. The space complexity of our Turing machine simulation is dominated by this. 

For the time analysis, first note that  generating the configurations consists of running a counter that takes expected time $O(\log^2 n)$, see~\cite{nubots} for details. After the configurations are generated, each iteration of matrix multiplication, addition, and deletions is bounded by time $O(\log^3 n)$. Since we do $O(\log n)$ matrix multiplications and additions the expected time is~$O(\log^4 n)$ using nubots.  The expected time of the  other rearrangements and computations  during the construction  is dominated by that of matrix multiplication. 
This completes the proof  of Lemma~\ref{lem:LTMtonubots}. 
\end{proof}

\subsection{Generating and simulating a Boolean circuit: proof of Theorem~\ref{thm:circuitsim}}\label{sec:CirGenSim}
We define the nubot monomer encoding of gates and  Boolean circuits. Boolean circuits were defined in Section~\ref{sec:circuitsDefs}.

\begin{definition}[Nubot monomer encoding of a  gate]\label{def:nubotgate}
 The encoding $\nb{g}$ of a Boolean circuit gate $g$ is as follows: a single gate monomer encodes the gate type (AND, OR, or NOT), directly above the gate monomer are $k$ line segments of monomers called result segments where $k$ is the gate's out-degree. Each result line segment encodes a destination gate number using $O(\log n)$ binary monomers, where~$n$ is the circuit size. There is an empty region of height $O(\log n)$ below the gate monomer  called the input region (Figure~\ref{fig:circuitsim}(b): dotted blue regions). 
 \end{definition}

A gate is simulated as follows. The input region of $\nb{g}$ is an empty region to which a line of monomers, that encode the inputs to  $g$, can attach (Figure~\ref{fig:boolAndNubot}(d): dotted blue regions). Upon attachment of the input lines, the gate monomer computes $g$'s Boolean function. Let $g$ be a gate with out-degree 2, and which outputs to the gates $g_1$ and $g_2$. 
The result region of nubot gate $\nb{g}$  consists of two lines of binary monomers that encode the wires that lead to $g_1$ and $g_2$  (Figure~\ref{fig:boolAndNubot}(d): solid blue regions). The simulation of wires is covered in the proof below. 

\begin{definition}[Nubot monomer encoding of a Boolean circuit]\label{def:nubotcircuit}
A  Boolean circuit~$c$ is encoded as a nubot configuration consisting of the encoded gates (Definition~\ref{def:nubotgate}) written in layers, one for each each layer in $c$ (see Figure~\ref{fig:boolAndNubot}).
 Within a layer, the encoded gates are horizontally spaced apart by the circuit size.\end{definition}

\begin{proof}[Proof  of Theorem~\ref{thm:circuitsim}] 
The proof has two parts, circuit generation and circuit simulation. 

\subsubsection {Circuit generation}
Let $c_n \in  \mathcal{C} $ be a  Boolean circuit with $n$ input gates that we wish to simulate. From the theorem statement  $\mathcal{C}$  is uniform by logspace Turing machine $\mathcal{M}$.   To generate the encoding of $c_n$ as nubot monomers, first,  $\mathcal{M}$ on input~$1^n$ is simulated via Lemma \ref{lem:LTMtonubots} to give a line of monomers $\lineseg{\nb{c_n}}{}{}$ that encodes $c_n  = \mathcal{M}(1^n)$. Next,  this ``linear'' encoding of $c_n$ geometrically unfolds into a two-dimensional ``ladder'' format, with one encoded circuit  layer per rung, as  shown in Figure~\ref{fig:boolAndNubot}(c) and  defined in Definition~\ref{def:nubotcircuit}. We use the folding technique from~\cite{nubots} that takes expected time $O(\log^2 \ell)$ to fold a length $\ell$ line into a square (here we modify the technique to fold a line into a comb, then on the teeth of the comb the gate result monomers fold out from each of the teeth to give the structure in Figure~\ref{fig:boolAndNubot}).
Since~$|\lineseg{\nb{c_n}}{}{}|$ is  polynomial  in $n$, the rearrangement happens in expected time~$O(\log^ 2 n)$.

\begin{figure}[t]
\begin{center}
    \includegraphics[width=0.65\textwidth]{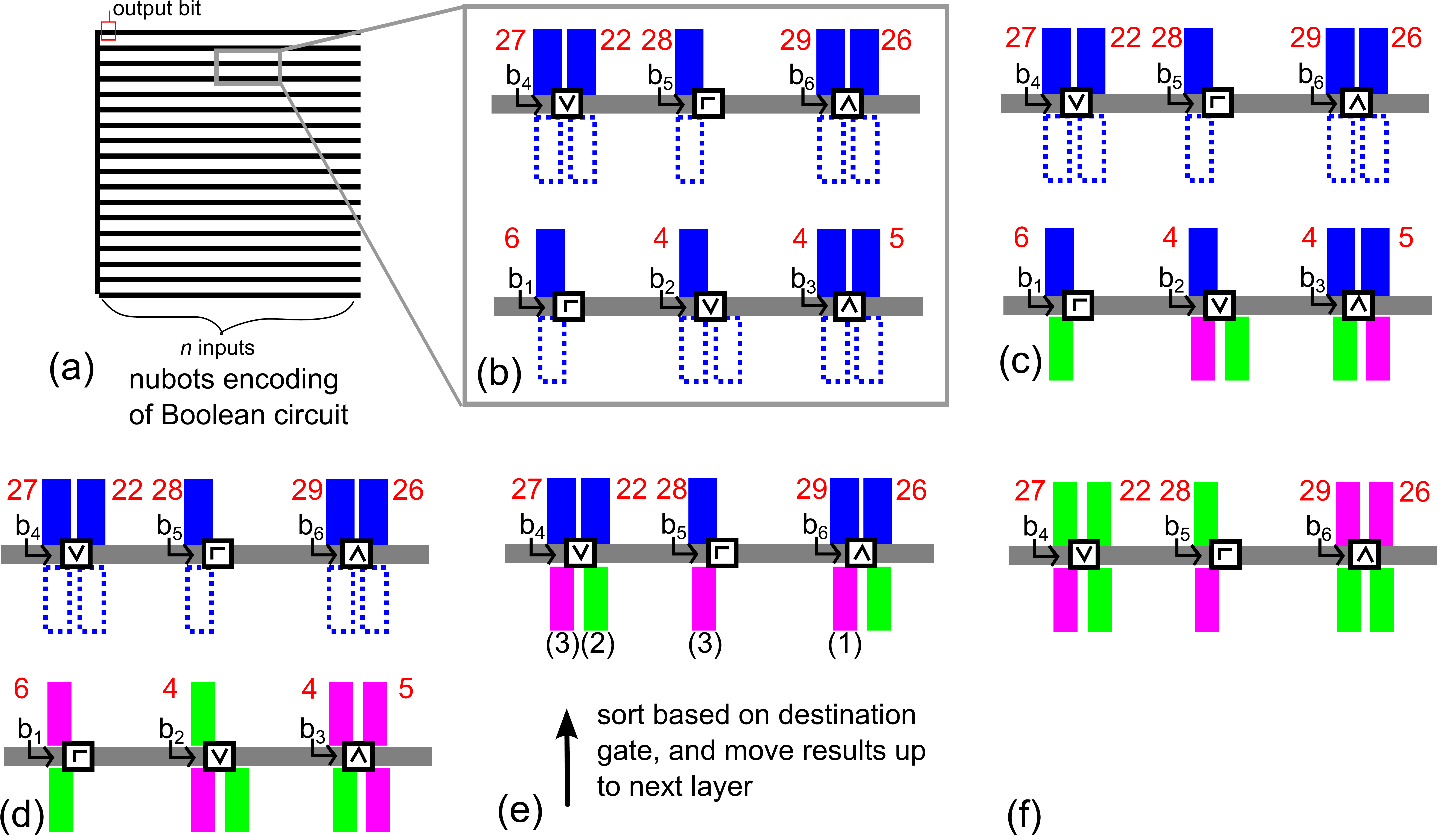} 
\end{center}
  \caption{Nubots configuration that simulates a Boolean circuit. Dotted blue regions denote gate input regions. Blue rectangles denote gate result regions. The gate address encoded by a line segment is in red. Gates are monomers. Pink encodes 0, green encodes 1. (a) Initial configuration that encodes the entire circuit and its input. (b) Zoom-in shown the initial configuration of 6 gates. (c) Input line segments arrive at the gate. (d) The gate monomers calculate their respective gates' result, and communicate the result bit to the result segments (shown as colour change).  (e) The input monomers and grey line are deleted. The result monomers sort themselves using the sorting algorithm (not shown) and are moved upwards to be the input to the next layer. The gate from which each input monomer originates is shown in parenthesis (these values came from gates 3, 2, and 1---not shown). (f) The gate monomers of the next layer calculate their result values.  }
  \label{fig:circuitsim}
\end{figure}

\subsubsection{Circuit simulation}
We have a nubot configuration that encodes a circuit as shown in Figure~\ref{fig:circuitsim}(a). 
We evaluate the encoded  Boolean circuit layer by layer, from input layer (bottom) to output layer (top), with each layer being evaluated in parallel. Evaluating a layer is a 3-step process shown in Figures~\ref{fig:circuitsim}(c), (d) and (e). First assume we have the  configuration shown in Figure~\ref{fig:circuitsim}(c): along the bottom grey line there are gate monomers (that encode AND, OR, or NOT), and below the bottom  grey line there are pink and green line segments of monomers that encode gate input bits 0 and 1 respectively, and above are blue line segments that encode the wires out of these gates (i.e. destination addresses to  the next layer above). Evaluation of a gate is straightforward: since the gates have fan-in $\leq 2$, the gate simply reads the  1 or 2 pink/green line segments below by reading its lower neighbours' states. The gate monomer computes the encoded gate's Boolean function and passes the resulting bit to the result line segments above. Note that  gates may have fan-out (or out-degree) as large as the circuit size, i.e. polynomial in input length, so for this we assume  adjacent gates on a layer are spaced at least as far apart horizontally as the circuit size. Then a gate communicates its result to all of them via shift or lift synchronization (in expected time logarithmic of circuit size).  After all gate monomers in a layer have completed this process, they synchronize. By now we have reached Figure~\ref{fig:circuitsim}(d). 

Boolean circuits may be non-planar and so have crossing wires when drawn in 2D, hence na\"ively moving bit-encoding monomers in the plane to the next layer above may cause unintended collisions.  We resolve this problem using our nubot sorting algorithm from Section~\ref{sec:sorting}. After  a layer has  synchronized, the (blue) line segments in the gate result regions of that entire layer are organised into a horizontal line, to serve as input to our sorting procedure. These gate result regions are then sorted by increasing wire number. The gates on the next layer above are assumed to be encoded in increasing (gate index) order. After sorting, the gate result regions are aligned with the gate above them (using counters) and pushed vertically upwards to the relevant gate. This is done in such a way that when it is finished there are no monomers below the new layer (any excess monomers are deleted); this deletion leaves enough space below for the sorting algorithm on the next iteration. 

After the monomer that encodes the circuit's unique output gate computes its result bit, it destroys itself, leaving a single monomer encoding the  output bit.  No  rules are applicable and so the system has halted with its answer.

\paragraph {Time, state, and space analysis of circuit simulation.}
There are $O(1)$ gate types, and all numbers are written using binary monomers. All other parts of the construction from previous sections  use $O(1)$  states. By stepping through the simulation with this in mind it is straightforward to obtain a state complexity of~$O(1)$.

We are simulating a circuit of size $O(n^k)$ and  depth $O(\log^j n)$. Each  layer of the circuit is encoded as a  monomer layer of height $O(\log n)$, giving a total height of of $O(\log^{j+1} n)$ for the encoded circuit. 
The width of an encoded layer is  $O(n^{2k} )$  which comes from the circuit size being $O(n^k )$, and from each gate being horizontally separated by a further $O(n^k )$ to handle fan-out (note a horizontal separation of a mere $O(\log n)$ monomers is sufficient for the sorting algorithm).  This gives $O(n^{2k}) \times O(\log^{j+1} n)$ space to lay out the circuit. However, sorting $n$ numbers, each written as length $O(\log n)$ bit strings, takes space $O(n \log n) \times O(n)$. 
Thus the total space complexity for the circuit simulation is $O(n^{2k}) \times O(n)$.  Lemma~\ref{lem:LTMtonubots} tells us that circuit generation takes space $n^{O(1)} \times n^{O(1)}$ (the hidden constants are coming from the logspace bounded Turing machine), which dominates the total space for both circuit generation and circuit simulation. 

The asymptotically slowest part of  simulating a circuit layer is the sorting algorithm, which takes expected time $O(\log^3 n)$ per layer. There are $O(\log^j n)$ layers, thus the total expected time for the simulation is  $O(\log^{j+3} n)$. The circuit generation takes time $O(\log^{4} n)$ from Lemma~\ref{lem:LTMtonubots}, but since we assumed that $j \geq 1$ in the statement of Theorem~\ref{thm:circuitsim}, this leaves the total expected time for both circuit generation and circuit simulation at $O(\log^{j+3 }n)$.
This completes the proof  of Theorem~\ref{thm:circuitsim}. 
\end{proof}

\section*{Acknowledgments.}
We thank Erik Winfree for valuable discussion and suggestions on our results,  Paul Rothemund for stimulating conversations on molecular muscle, Niall Murphy for informative discussions on circuit complexity theory, and Dhiraj Holden and Dave Doty for useful discussions. Damien thanks Beverley Henley for introducing him to developmental biology many moons ago. 

\small

\end{document}